\documentclass[10pt,journal,compsoc]{IEEEtran}
\usepackage{amsmath}
\usepackage{color}
\usepackage{amssymb}
\usepackage{verbatim}
\usepackage{enumitem}
\usepackage{bbm}

\ifCLASSINFOpdf 
\usepackage[pdftex]{graphicx}  
\usepackage[pdftex,colorlinks,
           bookmarks=true,%
                     ]{hyperref}
\usepackage[numbers,sort&compress]{natbib} 
\usepackage{hypernat} 
\else 
\usepackage[dvips]{graphicx}

\usepackage[dvips, colorlinks, bookmarks=true, ]{hyperref}

 \usepackage[numbers,sort&compress]{natbib} 
\fi

\hypersetup{
   bookmarksnumbered,
   pdfstartview={FitH},
   citecolor={blue},
   linkcolor={red},
   urlcolor=[rgb]{0,0.55,0},
   pdfpagemode={UseOutlines},
   pdftitle={...}, 
   pdfauthor={A.\ Chattopadhyay,\ B.\ Blaszczyszyn,\ E.\ Altman}
}
 
\graphicspath{{./../figures/}{./../plots/}}

\IEEEoverridecommandlockouts

\newcommand{\remove}[1]{}

\newcommand{\ind}{\mathbf{1}}
\newcommand{\E}{\mathbf{E}}
\newcommand{\Pro}{\mathbf{P}}

\newtheorem{theorem}{Theorem}
\newtheorem{corollary}{Corollary}
\newtheorem{lemma}{Lemma}

\newtheorem{definition}{Definition}

\newtheorem{algorithm}{Algorithm}
\newtheorem{assumption}{Assumption}

\newcommand{\qed}{\hfill \ensuremath{\Box}}
\usepackage{setspace}
\usepackage{subfigure}

\begin{document}

 \title{Location Aware Opportunistic Bandwidth Sharing between Static and Mobile Users with Stochastic Learning in Cellular Networks
 \thanks{Arpan Chattopadhyay is with Electrical Engineering department, University of Southern California, Los Angeles, USA.  Bart{\l}omiej B{\l}aszczyszyn are with Inria/ENS, Paris, France. Eitan Altman is 
 with Inria, Sophia-Antipolis, France. Email: achattop@usc.edu, Bartek.Blaszczyszyn@ens.fr,
 eitan.altman@sophia.inria.fr . \newline
 This work was done when Arpan Chattopadhyay was a postdoctoral researcher in Inria/ENS, Paris, France. \newline 
 {\bf                 All appendices are provided in the supplementary material.} }
}

\author{
Arpan~Chattopadhyay, Bart{\l}omiej~B{\l}aszczyszyn, Eitan~Altman \\
}

\thispagestyle{empty}

\IEEEcompsoctitleabstractindextext{
\begin{abstract}
In this paper, we consider the problem of location-dependent opportunistic bandwidth sharing between static and 
mobile (i.e., moving) downlink users in a cellular network. Each cell of the network has some fixed number of static users. 
Mobile users enter the cell, 
move inside the cell for some time and then leave the cell. In order to provide higher data rate to the highly mobile users whose     fast fading channel variation is difficult to track, we propose location dependent 
bandwidth sharing between the two classes of static and mobile users; the idea is to provide higher bandwidth to the mobile users at 
favourable    locations, 
and provide higher bandwidth to the static users in other times. Our approach is agnostic to the way the bandwidth is further shared 
within the same class of users; it can be combined with any particular bandwidth allocation policy employed for one of these two 
classes of users. 
We formulate the problem as a long run average reward Markov decision process (MDP) where the per-step reward is a linear combination 
of instantaneous data volumes received by static and mobile users,  
and find the optimal policy. The optimal policy is binary in nature; it allocates the entire bandwidth either to the static users or to the mobile users at any given time. 
The reward structure of this MDP is not  known in general, and it may change with time. To alleviate these issues, 
we propose a learning algorithm based on single timescale {\em stochastic approximation}. 
Also, noting that the MDP problem can be used to maximize the long run average data rate for mobile users subject to a constraint 
on the long run average data rate of static users, we provide a learning algorithm based on 
multi-timescale stochastic approximation. 
We  prove asymptotic convergence of the bandwidth sharing policies under these learning algorithms to the optimal policy. 
The results are extended to address the issue of fair bandwidth sharing between the two classes of static and mobile users, where the notion of 
fairness is motivated by the popular notion of $\alpha$-fairness in the literature. Numerical results exhibit 
significant performance improvement by our scheme, as well as fast convergence, and also demonstrate the trade-off between 
performance gain and fairness requirement.
\end{abstract}

}

\maketitle

\section{Introduction}\label{section:introduction}
In recent years, cellular traffic has shown unprecedented growth, due to the proliferation 
of high-specification handheld/mobile devices such as smartphones and tablets. It is speculated that increasing use of 
applications such as video streaming or downloading, image or 
media file transfer, social networking applications and cloud services (requested or run by these devices)  will further increase 
this traffic demand in the coming years. 
In order to meet the enormous bandwidth demand for these applications, the use of macro-assisted small cell networks 
(see \cite{andrews-etal12femtocells}, \cite{pauli-etalYYhet-LTE-ICI-coordination}, 
\cite{stanze-weber13het-LTE-A}, \cite{nakamura-etal13trends-small-cell-LTE}, \cite{ishii-etal12novel-architecture-LTE-B}) have 
recently become popular; the small cells 
(such as femtocells and picocells) can meet the bandwidth demand of the users, while the macro base stations are supposed to 
provide cellular coverage.

While small cell networks can provide high throughput to the static users,
the performance of mobile users (i.e., fast moving users) deteriorates due to frequent 
handoff at  cell boundaries  resulting in huge signaling overhead 
(see \cite{camp-etal02survey-mobility-models}) and temporary data outage for each handoff. 
As a solution to this problem, the use of heterogeneous network architecture has been proposed 
(see \cite{chattopadhyay-etal16mobility-work-arxiv}), where only macro base stations can serve the 
mobile users; the relatively large cell size of the macro base stations result in a much smaller handoff rate for mobile users 
in this architecture. Static users can be served by either macro or micro base stations. However, this alone is not capable of 
meeting the growing traffic demand from mobile users,  
and hence new improvements in PHY and MAC techniques are essential. 

In order to address the above issue, we propose opportunistic (and dynamic) sharing 
of the total allocated bandwidth to a base station, 
by the two classes of static and mobile {\em downlink} users, based on user locations.  
The  transmission bandwidth available for a macro base station can be shared among its users in many ways. 
However, the interference field and downlink path-loss vary over various locations inside a macro cell, 
due to spatio-temporal variation in fast fading, distance and shadowing  from various 
interfering base stations. Hence, a natural 
way to improve user throughput is to employ dynamic bandwidth sharing among the static and mobile users inside a macro cell, depending on their  
instantaneous location, direction of motion and speed; the idea is to provide more bandwidth to the mobile users opportunistically 
when they are at favourable locations, in a distributed fashion so that the base stations need not communicate among themselves to decide 
on bandwidth allocation. This approach also alleviates the problem of measuring fast fading channel variations from the base station to the highly mobile users. Our goal is to maximize the time average of a linear combination of the expected data rates of mobile and static users. 
We formulate the problem as an average reward Markov decision process (MDP), and establish the policy structure. However, the
decision making requires information on the location of other base stations 
and the shadowing realizations from other base stations to various locations in the macro cell; 
these quantities might not be known to the macro base station, and some of them might even change over time. Hence, 
instantaneous data rate for a fast moving mobile user may not be computable in the presence of the spatially varying 
unknown interference field; only the cumulative amount 
of data downloaded by the mobile user over an interval will be available to the macro BS. Hence, we provide a 
learning algorithm using the theory of 
stochastic approximation, and prove its asymptotic convergence to the optimal dynamic bandwidth sharing policy. 
Next, we  propose a learning algorithm based on  
multi-timescale stochastic approximation, which converges to the optimal policy 
for the problem of maximizing the time-average expected data rate 
of mobile users subject to a constraint on the time-average expected data rate of static users. Hence, the learning 
algorithms can be used by the macrocells to dynamically adapt the bandwidth 
sharing policy depending on mobile user locations. We also explain how to adapt the dynamic bandwidth sharing 
scheme when  fair bandwidth sharing between the classes of static and mobile (i.e., moving) users is required. Finally, 
we demonstrate numerically that the proposed dynamic (opportunistic) bandwidth allocation scheme can improve 
user performance significantly, and also demonstrate the trade-off between performance improvement and 
a measure of the degree of fairness in allocation.

\subsection{Related Work}\label{subsection:related_work}
There has been a vast literature on the impact of user mobility in wireless networks. The authors in 
\cite{grossglauser-tse02mobility-increases-capacity} have shown that mobility increases the capacity.  
\cite{bansal-liu03capacity-delay-mobility-ahn} has explored the trade-off between delay and throughput in ad-hoc networks in presence of mobility. 
The papers \cite{bonald-etal09flow-level-performance-mobility}, 
\cite{bonald-etal04mobility-flow-level-data-systems}, \cite{borst-etal12capacity-with-mobility},  \cite{karray-mobility}, 
\cite{orlik-rappaport01handoff-arrival-process-cellular} study the impact of inter and intra cell mobility on capacity, and also 
the trade-off between throughput and fairness; these results show that mobility increases the capacity of cellular networks when base stations 
cooperate among themselves.

However, in practice, base stations may not cooperate. Moreover, due to frequent handoff of fast moving mobile users, 
significant control bandwidth has to be dedicated for handoff management; it is often the case 
that handoff results in temporary data outage for mobile users. In order 
to optimize the performance of cellular networks under user mobility, we propose to use {\em optimal} 
dynamic bandwidth sharing between the two classes of static and mobile users (depending on user locations); this 
problem is formulated as an average reward MDP (where the reward is a time average linear combination of data rates of static and mobile users) 
and later learning algorithms for computation of the optimal policy are  
provided. There have been some work in the literature relevant to our paper.  The authors of \cite{bonald-mobility} also have evaluated 
gain in performance due to mobility by favouring users with good radio channel conditions, however they did not propose any optimal bandwidth allocation scheme among users under mobility when channel condition may not be measured accurately.  The paper \cite{kushner-whiting04convergence-proportional-fair} deals with proportional fair scheduling algorithm for a {\em fixed} population of users with time-varying channel conditions due to mobility; this paper proposes a {\em single-timescale} stochastic approximation based algorithm to estimate throughput of each user, and analyses its convergence. The paper \cite{margolies-etal16exploiting-mobility-cellular-scheduling} essentially considers location based proportional fair scheduling over a finite time horizon to a fixed user population, but it does not provide any convergence analysis of the proposed algorithm. The authors of \cite{margolies-etal16exploiting-mobility-cellular-scheduling} allocate bandwidth among users opportunistically via the construction of a spatial radio map which depends on path-loss and slow fading but is averaged over fast fading; in other words, they pursue a more experimental and data-driven approach. Reference \cite{ali-etal07mobility-assisted-opportunistic-scheduling} solves the problem of bandwidth allocation among users as a static optimization problem that yields the fraction of bandwidth to be allocated to each user at a given state; this method maximizes the sum throughput of users, is easy to implement, but it requires the knowledge of user mobility statistics at the base station.

{\em However, to the best of our knowledge, there has been no prior work that considers  optimal 
dynamic bandwidth sharing depending on user location in order to maximize the sum data rate of mobile users subject to a constraint on the time-average sum data rate of static users, and proposes any learning algorithm based on multi-timescale stochastic approximation for this problem 
with provable convergence guarantee; in our current paper, we seek to address these problems.}

\subsection{Organization and Our Contribution}\label{subsection:our_contribution}
The rest of the paper is organized as follows:
\begin{itemize}
\item The system model is described in Section~\ref{section:system_model_and_notation}.
 \item  In Section~\ref{section:mdp-formulation-for-the-unconstrained-problem}, we develop optimal bandwidth 
 sharing strategy between the two classes of static and mobile users in a {\em single} cell, so as to maximize 
 the time average of a linear combination of expected sum throughputs of static and mobile users inside the cell. 
 This unconstrained optimization problem is formulated as an average reward Markov decision process (MDP), and 
 optimal policy structure is derived analytically. To the best of our knowledge, this model and mathematical contributions including the specific policy structures are new contributions to the literature and they can be used in practical wireless cellular networks.
 \item In Section~\ref{section:learning-algorithm-given-lagrange-multiplier}, we provide 
 a learning algorithm based on stochastic approximation, which converges asymptotically to the optimal bandwidth sharing policy, without 
 using the transition and cost structure of the MDP. 
 \item Noting that the unconstrained optimization problem can be used to solve the 
 constrained problem of maximizing the time average sum rate of the mobile users subject 
 to a constraint on the time-average sum rate of the static users, we provide, in Section~\ref{section:learning-algorithm-constrained-problem}, a 
 learning algorithm based on multi-timescale stochastic approximation, which {\em provably} converges to the optimal policy   for the constrained 
 problem. 
 This multi-timescale stochastic approximation based learning algorithm yields a randomized policy, and the 
 randomization technique proposed in this paper is novel to the literature. This randomized bandwidth allocation technique can be used in practical cellular network where a precise radio map of the cell is not available.
 \item In Section~\ref{section:fairness}, we show how the dynamic (and opportunistic) bandwidth sharing schemes developed in previous 
 sections can be adapted to ensure a fair allocation between the two classes of static and mobile users; we specifically adapt the notion of $\alpha$-fairness and extend our algorithms to this framework. 
 \item In Section~\ref{section:numerical-work-performance-improvement}, we numerically demonstrate considerable  performance gain due to 
 opportunistic bandwidth sharing, and also explore the trade-off between performance gain and fairness in allocation. Fast convergence of the proposed learning algorithm is also demonstrated.
 \item In Section~\ref{section:implementation-issues}, we show the equivalence of the global problem of 
 decentralized maximization of the time average of a linear combination of the expected sum throughputs of 
 mobile and static users, with a problem where each 
 base station seeks to maximize the time average of a linear combination of the expected sum rates of all mobile users and 
 all static users via location aware opportunistic bandwidth sharing between the two classes of static 
 and mobile users. We also explain how to modify our algorithms in case the location of users in a cell 
 are not known perfectly, thereby extending the algorithms to a more practical regime. We also motivate the need for location based bandwidth sharing instead of channel estimation based bandwidth allocation. It has also been argued how to extend the proposed algorithms for more generalised system model.
 \item Finally, we conclude in Section~\ref{section:conclusion}. 
 \item All proofs are provided in the appendix.
\end{itemize}

\section{System Model}\label{section:system_model_and_notation}

We consider a cellular network with multiple (possibly infinite and heterogeneous) 
base stations (BSs) on the two dimensional plane. Among these BSs, we consider one {\em single} BS and focus on the cell 
served by that BS (see Figure~\ref{fig:location-dependent-bandwidth-allocation}); this BS can be a macro BS if the network is heterogeneous. 
We consider two classes of {\em downlink}  
users served by this BS: {\em Static users (SU)} and {\em mobile users (MU)}. We assume that there exist 
multiple directed lines/routes (e.g., roads) crossing the cell, and MUs  are moving along these lines with
constant speed $v$. 
This can be a model for the roads in urban or suburban areas where users sitting in 
fast moving cars download contents from the base stations. 
Given a realization of the  line segments inside the cell, we assume that, MUs are entering a cell along each line according 
to a time-homogeneous process, and the arrival rate is potentially different along different routes. 
We assume that all the base stations transmit simultaneously (either on the same band or using frequency reuse), and 
these transmissions create interference at the SUs and MUs.

\begin{figure}[!t]
\begin{center}
\includegraphics[height=5cm, width=7cm]{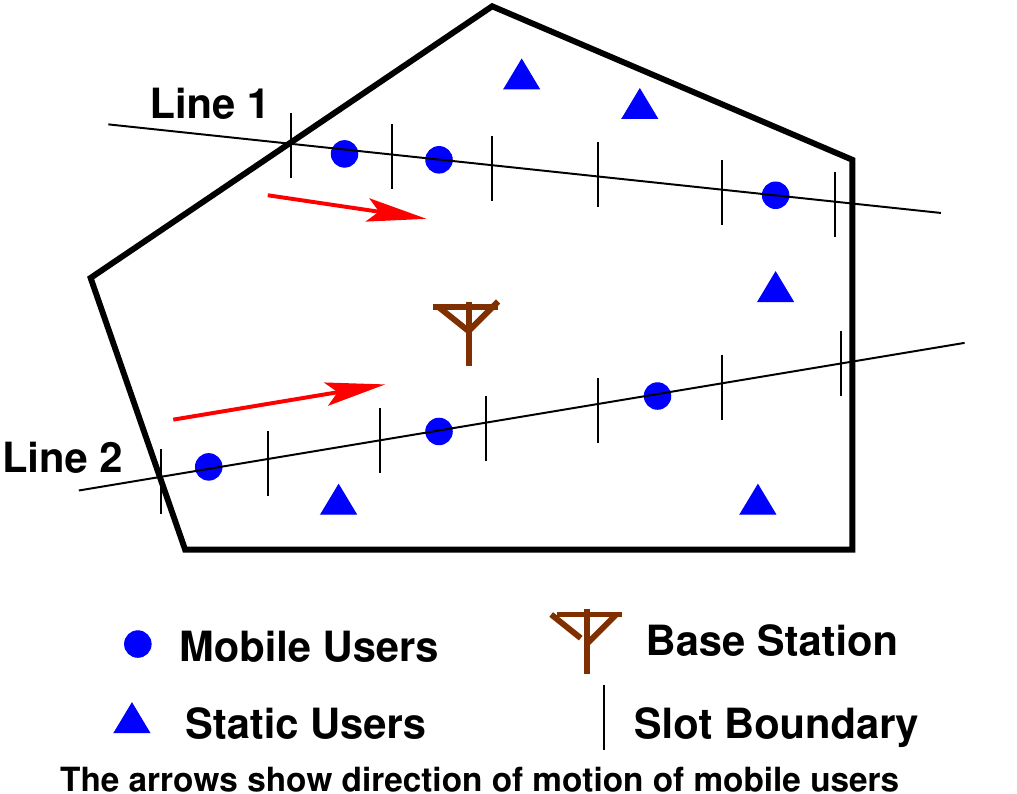}
\end{center}
\caption{A snapshot of one cell with base station, static users and mobile users. Two lines (line~$1$ and line~$2$) cross the cell; these 
lines have lengths $l_1 v \sigma$ and $l_2 v \sigma$ respectively inside the cell, where $l_1=5$, $l_2=6$, $v$ is the 
velocity of mobile users and $\sigma$ is the time slot duration. Mobile users entering the cell stay for 
$l_1$ or $l_2$ slots and then leave the cell. In each slot, depending on the instantaneous location of all users inside the cell, bandwidth is 
shared opportunistically between the classes of static and mobile users.}
\label{fig:location-dependent-bandwidth-allocation}
\end{figure}

In order to mathematically formulate the dynamic bandwidth sharing problem, 
we make the following simplified modeling assumptions (also, see Figure~\ref{fig:location-dependent-bandwidth-allocation} for 
a clear pictorial description):
\begin{itemize}
\item Time is discretized into slots of duration $\sigma$. Hence, a MU moves $v \sigma$ distance in one slot.
\item The BS  under consideration knows the locations of all static users associated with it.
\item The  BS knows the lines intersecting with its  cell, and also the lengths of these corresponding line segments. 
Let us denote the number of line segments intersecting with 
the cell under consideration, by $n$. Let the $i$-th line segment have  
length $l_i\sigma v$, i.e., a MU can remain in the 
$i$-th line segment for $l_i$ number of time slots. 
Thus, the model is as follows: any MU that enters the cell 
(after coming out of a handoff) along the $i$-th line segment spends 
a time of $l_i$ slots, and finally 
enters another cell. The values $\{l_i\}_{1 \leq i \leq n}$ are known to the base station. 
\item We allow the MU 
arrival rates along the $n$ line segments to be unequal. We denote the number of arrivals to the cell along line $i$ 
at time slot $t$ by a {\em bounded} random variable $A_{i,t}$; we assume that $A_{i,t}$ is i.i.d. across  $t$ and independent across $i$;  
the arrival process will be correlated across cells, 
but that does not affect the resource allocation problem 
for a single cell. 
\item At the beginning of each slot $\tau$, the BS under consideration decides the fraction $\eta_{\tau}$ of 
the available bandwidth to be dedicated for transmission 
to the mobile users. The remaining bandwidth is assigned to the set of static users. 
{\em In each slot, 
a base station can allocate equal bandwidth to all available mobile users, or possibly unequal (arbitrarily) 
bandwidth sharing among the mobile users is done.} Similarly, the $(1-\eta_{\tau})$ 
fraction of bandwidth  can be shared arbitrarily among all static users. {\em In this paper, we assume equal bandwidth sharing within one user class for the sake of illustration.}

It is important to note that, for fast moving users, 
traditional channel estimation may not be very accurate since the user might 
travel the fading coherence distance very fast; hence, dynamic bandwidth allocation among 
users based on instantaneous channel qualities  
may not be feasible. This necessitates location dependent bandwidth sharing which works on a slower timescale compared 
to variation in fading due to high speed of users. 
However, our scheme of sharing bandwidth between two classes of users can well accommodate any 
scheduling policy employed within the same class of users.  
Another reason for not considering location-dependent (resp., channel quality based) bandwidth allocation to individual users 
(instead of user classes) 
is that this will result in allocation of bandwidth to the user having the best location (resp., channel quality) at any 
given time slot, which might be unfair to all other users. {\em See Section~\ref{subsection:channel-estimation-versus-location-dependent} for detailed discussion on the necessity of location-dependent bandwidth sharing instead of channel quality measurement based bandwidth allocation.}

\item At the beginning of each slot, the base station gets to know the number $m$ of existing mobile users inside the cell 
(including the newly arrived MU), the index set $z_1,z_2,\cdots,z_m$ 
of lines on which each of these mobile users are moving, and also the remaining sojourn times 
(in terms of slots) $t_1,t_2,\cdots,t_m$ of those mobile users. 
This can be done via the GPS connection of the mobile users. Otherwise, since the base station records the 
time and line of entry of a new MU into the cell, and since the velocity is known, 
the base station can always calculate the location of any mobile 
station inside the cell. We define $s:=(\{t_i, z_i\}_{i=1}^m)$ to be the state of the system at the beginning 
of a slot.

{\em We will explain in Section~\ref{subsection:error-in-location-estimation} 
how we can relax the assumption on availability of perfect information of the system state to the decision maker.}

\item At state $s$, if all available bandwidth (assumed to be equal to $1$~unit) is allocated only to MUs, 
then, given a bandwidth 
sharing scheme among all SUs and a bandwidth sharing scheme among all MUs, and given the 
realization of shadowing and path-loss from each BS to each location in the cell, 
the amount of data each MU will be able to download over a slot is a random variable since the fading 
process seen by each user (from the serving BS and interfering BSs) over this slot is random. However, if these 
quantities and the fading distribution is known, 
the base station can calculate the expected data volume each user will be able to download until the beginning of the 
next slot.\footnote{Note that, 
for a given realization of the location of base stations and for a 
given realization of the spatially varying shadowing process, 
the amount of interference at any location is not a random variable if there is no fading. 
Even the interference averaged over random, time-varying 
fading is a 
deterministic quantity.  But this quantity is unknown in general to a 
BS which does not possess global information about the base station 
locations and shadowing process.}. Let us define 
$R_{mobile}(s)$ to be the (random) total amount of data the MUs download {\bf per unit bandwidth} if the entire 
bandwidth is allocated to MUs, and similar 
meaning applies for $R_{static}(s)$ (i.e., this is the random amount of data the SUs can download in a slot 
in case the entire bandwidth is allocated to SUs). 
In presence of fading, the expectations of these two random variables (expectation taken over fading distribution) 
are denoted by $\overline{R}_{mobile}(s)$ and $\overline{R}_{static}(s)$. 
Note that, $\overline{R}_{mobile}(s)$ and 
$\overline{R}_{static}(s)$ are dependent on the shadowing realizations from all base stations to the static and mobile 
users over various locations in the cell (since they will determine the signal to interference 
ratio for various users at different locations).

We will assume in Section~\ref{section:mdp-formulation-for-the-unconstrained-problem} that 
$\overline{R}_{mobile}(s)$ and $\overline{R}_{static}(s)$ are known to the BS; this assumption will be 
relaxed in subsequent sections.

\end{itemize}

\section{Opportunistic Bandwidth Allocation  Under Perfect Knowledge of Mean User Rates $\overline{R}_{mobile}(s)$ and $\overline{R}_{static}(s)$}
\label{section:mdp-formulation-for-the-unconstrained-problem}

\subsection{Markov decision process formulation}
\label{subsection:smdp-formulation}
We formulate the dynamic (i.e., opportunistic) bandwidth allocation problem for a BS as a Markov decision process. 
We assume in this section that the base station knows $\overline{R}_{mobile}(s)$ and 
$\overline{R}_{static}(s)$ perfectly for each state $s$ at each slot. 

\subsubsection{State Space} 
New mobile users arrive to the cell in each slot. The state of the system 
at the beginning of a slot is considered. The state at the beginning of a slot (after new arrivals in 
the previous slot) is of the form 
$s:=(\{t_i, z_i\}_{i=1}^m)$, where $m$ is the number of mobile users present in the cell, $t_i$ is the residual 
sojourn time of the $i$-th user in the cell, and $z_i \in \{1,2,\cdots,n\}$ 
is the index of the line along which the $i$-th mobile user is moving; if $z_i=k$, then $t_i \in \{1,2,\cdots,l_k\}$. 
Note that the state space is finite since the number of arrivals in each slot is bounded and each mobile user stays inside the 
cell for a bounded number of slots.

\subsubsection{Action Space}  
Given a state, the  BS takes an action $x \in [0,1]$; $x$ is the fraction of bandwidth 
the  BS decides to allocate to the mobile users. Hence, our action space is $[0,1]$. In this work, we assume that,  
at any given time, all static users share  the $(1-x)$ fraction of bandwidth equally among themselves, and all 
mobile users share the $x$ fraction of bandwidth equally among them.\footnote{From the optimization point of view, 
it will always be better to allocate $x$ fraction of bandwidth to the {\em best} mobile user at a given slot, and 
$(1-x)$ fraction to the {\em best} static user for ever. But this will result in complete starvation for many 
static users, and short-term service unfairness among the mobile users; each mobile users will get high data rate in some slots, and 
very low (possibly zero) data rate in some other slots.} 

\subsubsection{State Transition} 
For current state $s:=(\{t_i, z_i\}_{i=1}^m)$, if $p$ MUs arrive to the cell in a slot, then the next state will be 
$s'=(\{(t_i-1)^+, z_i\}_{i=1}^m, \{t_i, z_i\}_{i=m+1}^{m+p})$, where 
$z_i \in \{1,2,\cdots,n\} \forall i \in \{m+1,m+2,\cdots,m+p\}$ is the index of the line along which the $i$-th 
new arrival at the slot enters the cell, and  $t_i=l_k$ if $z_i=k$ for $i \in \{m+1,m+2,\cdots,m+p\}$. 
In course of this, if $(t_i-1)^+=0$ for any $i \in \{1,2,\cdots,m\}$, then information of that user is removed 
from the state since he has already left the cell. We denote the state at time slot $\tau$ by $s(\tau)$.

\subsubsection{Policy} 
A stationary policy $\eta(\cdot | \cdot)$ is a family of  probability distributions $\eta(\cdot | \cdot)$ on the action space 
$[0,1]$ conditioned on the state $s$; i.e., $\eta(\cdot | s)$ denotes the probability distribution 
of the action taken whenever the system reaches state $s$. If $\eta(\cdot | s)$ is such that for each state $s$, 
the policy chooses one action with probability $1$, then the policy is called a stationary deterministic policy 
$\eta(\cdot)$; in this case, 
$\eta(s)$ denotes the action taken at state $s$. 
We denote by $\eta_{\tau}$ the action taken at time $\tau$ (i.e., the fraction of bandwidth allocated 
to the class of MUs in slot $\tau$); this will be equal to $\eta(s(\tau))$ if a stationary deterministic policy $\eta(\cdot)$ 
is used in decision making. We denote by $\eta_{\tau}$ a number 
in $[0,1]$, and by $\eta(\cdot)$ a function.

\subsubsection{Single Stage Reward} 
If the system state is $s(\tau)$ at slot $\tau$, and if an action $\eta_{\tau} \in [0,1]$ is taken, 
the total (random) reward for 
the base station at decision epoch $\tau$ is defined as 
$$R(\tau):=\eta_{\tau}R_{mobile}(s(\tau))+\xi 
(1-\eta_{\tau})R_{static}(s(\tau)).$$

\subsubsection{Objective Function}

Let us denote the expectation under policy $\eta(\cdot | \cdot)$ by 
$\E_{\eta(\cdot | \cdot)}$; the expectation is over the 
randomness in the policy and over the randomness in state evolution. 
We seek to solve the following problem of maximizing the time average of the expected reward 
per slot:

\tiny
\begin{eqnarray}
 && \sup_{\eta(\cdot | \cdot)} \liminf_{N \rightarrow \infty}\frac{1}{N}\sum_{\tau=1}^{N} \E _{\eta(\cdot | \cdot)} \bigg( \eta_{\tau}\overline{R}_{mobile}(s(\tau))+\xi 
(1-\eta_{\tau})\overline{R}_{static}(s(\tau)) \bigg) \nonumber\\
&& \label{eqn:unconstrained_mdp_for_a_cell}
\end{eqnarray}
\normalsize

Here $\xi \geq 0$ can be considered as a Lagrange multiplier; it captures the emphasis we put 
on the time average sum throughput of SUs and MUs in the objective function. 
This problem is an unconstrained optimization problem. 

Note that, there are two expectations in this objective function: one is over randomness in the fading process 
(which are captured by $\overline{R}_{mobile}(s)$ and 
$\overline{R}_{static}(s)$), and the other one is over the randomness in the policy  and over the 
randomness in the state evolution (captured by $\E_{\eta(\cdot | \cdot)}$).

The problem \eqref{eqn:unconstrained_mdp_for_a_cell} has a stationary, deterministic optimal 
policy (by standard MDP theory), which we denote by $\eta_{\xi}^*(\cdot)$. 
Under the deterministic policy $\eta_{\xi}^*(\cdot)$, the optimal action at state $s$ is denoted by 
$\eta_{\xi}^*(s)$ (parametrized by $\xi$) or simply by $\eta^*(s)$. The optimal value for the objective in 
\eqref{eqn:unconstrained_mdp_for_a_cell} is denoted by $\lambda^*(\xi)$ or simply by $\lambda^*$.

It has to be noted that, under $\eta_{\xi}^*(\cdot)$, we have 
$\lim_{t \rightarrow \infty} \frac{\sum_{\tau=1}^t R(\tau)}{t}=\lambda^*(\xi)$ almost surely 
(by the ergodicity of the Markov chain $\{s(\tau)\}_{\tau \geq 1}$). 

Later in Section~\ref{subsection:connection-cell-level-global-problem}, we relate 
\eqref{eqn:unconstrained_mdp_for_a_cell} to a global optimization problem over multiple
cells.

\subsubsection{Connection Between the Unconstrained Problem and a Constrained Problem}
\label{subsubsection:connection-constrained-unconstrained}
The unconstrained optimization 
problem~(\ref{eqn:unconstrained_mdp_for_a_cell}) can be used to solve 
the following constrained optimization problem of maximizing the time-average sum data rate for 
the mobile users while satisfying a minimum time-average sum data rate constraint $R_0$ for static users:

\footnotesize
\begin{eqnarray}
 &&\sup_{\eta(\cdot | \cdot)} \liminf_{N \rightarrow \infty}\frac{1}{N}\sum_{\tau=1}^{N} \E _{\eta(\cdot | \cdot)} \bigg( \eta_{\tau}\overline{R}_{mobile}(s(\tau)) \bigg) \nonumber\\
 &s.t.,&  \liminf_{N \rightarrow \infty}\frac{1}{N}\sum_{\tau=1}^{N} \E _{\eta(\cdot | \cdot)} \bigg( (1-\eta_{\tau})\overline{R}_{static}(s(\tau)) \bigg) \geq R_0  \nonumber\\
 \label{eqn:constrained_mdp_for_a_cell}
\end{eqnarray}
\normalsize

It is well-known that by choosing an appropriate value $\xi^*$ for $\xi$ and solving the optimization 
problem~(\ref{eqn:unconstrained_mdp_for_a_cell}), one can find an optimal policy $\eta_{\xi^*}^*(\cdot | \cdot)$ 
for the constrained 
problem~(\ref{eqn:constrained_mdp_for_a_cell}) as well. 

The following standard result tells us how to choose the optimal {\em Lagrange multiplier} 
$\xi^*$ (see \cite[Theorem~$4.3$]{beutler-ross85optimal-policies-controlled-markov-chains-constraint}):

\begin{theorem}\label{theorem:how-to-choose-optimal-Lagrange-multiplier}
 Consider the constrained problem~(\ref{eqn:constrained_mdp_for_a_cell}). If there exists a multiplier 
 $\xi^* \geq 0$ and a policy $\eta_{\xi^*}^*(\cdot | \cdot)$ 
such that $\eta_{\xi^*}^*(\cdot | \cdot)$ is an optimal policy for the unconstrained problem~(\ref{eqn:unconstrained_mdp_for_a_cell}) 
under  $\xi^*$ and the constraint in 
(\ref{eqn:constrained_mdp_for_a_cell}) is met with equality under policy $\eta_{\xi^*}^*(\cdot | \cdot)$, 
then $\eta_{\xi^*}^*(\cdot | \cdot)$ is an optimal policy for the constrained problem~(\ref{eqn:constrained_mdp_for_a_cell}) also.\qed
\end{theorem}

{\em Remark:} We will see in Section~\ref{section:learning-algorithm-constrained-problem} that, 
in order to meet the constraint in \eqref{eqn:constrained_mdp_for_a_cell} with equality, 
we will need randomization between two deterministic policies (contrary to the 
fact that \eqref{eqn:unconstrained_mdp_for_a_cell} has a stationary, deterministic, optimal policy).

\subsection{Optimal Policy Structure}\label{subsection:policy-structure-and-computation}
In this section, we will only consider the unconstrained problem \eqref{eqn:unconstrained_mdp_for_a_cell}. 
We formulate the problem as a Markov decision process (MDP). The average reward optimality equation for this MDP is given by 
(see \cite[Chapter~$7$, Section~$4$]{bertsekas07dynamic-programming-optimal-control-1}): 

\footnotesize
\begin{eqnarray}
 h^*(s)&=&\max_{x \in [0,1]}  \bigg( x \overline{R}_{mobile}(s)+\xi (1-x) \overline{R}_{static}(s) \nonumber\\
 && -\lambda^*+\E (h^*(S')) \bigg)
 \label{eqn:Poisson-equation}
\end{eqnarray}
\normalsize
where $\lambda^*$ is the optimal average reward per slot for the problem~(\ref{eqn:unconstrained_mdp_for_a_cell}), 
$h^*(s)$ is the optimal differential cost at state $s$ (see 
\cite[Chapter~$7$, Section~$4$]{bertsekas07dynamic-programming-optimal-control-1} for thorough interpretation of the 
differential cost $h^*(s)$), and $S'$ is the (random) next state whose distribution depends 
on $s$ and the realization of new arrivals. Note that, state transition is independent of the action taken in any slot; 
hence, the expectation in $\E (h^*(S'))$ is taken only over the randomness in the new arrivals of MUs to the BS in one slot.

\begin{theorem}\label{theorem:policy-structure}
 {\em (Optimal policy $\eta_{\xi}^*(\cdot)$:)} If the state $s$ is such that, $ \overline{R}_{mobile}(s)-\xi \overline{R}_{static}(s) >0$, then optimal action is 
 $\eta_{\xi}^*(s)=1$. If $ \overline{R}_{mobile}(s)-\xi \overline{R}_{static}(s) <0$, then $\eta_{\xi}^*(s)=0$. 
 If $ \overline{R}_{mobile}(s)-\xi \overline{R}_{static}(s)=0$, then we can choose any action $\eta_{\xi}^*(s)$.
\end{theorem}
\begin{proof}
From \eqref{eqn:Poisson-equation}, we can say that:
 \begin{eqnarray*}
 && \eta_{\xi}^*=\arg \max_{x \in [0,1]}  \bigg( x \overline{R}_{mobile}(s)+\xi (1-x) \overline{R}_{static}(s) \\
 && -\lambda^*+\E (h^*(S')\bigg),
 \end{eqnarray*} 
 i.e., $\eta_{\xi}^*$ should be the maximizer in the average cost optimality equation. Since $\lambda^*$, $\xi$, $\overline{R}_{mobile}(s)$, $\overline{R}_{static}(s)$ and $\E (h^*(S'))$ are independent of $x$ in this optimization problem, we have 

$$\eta_{\xi}^*=\arg \max_{x \in [0,1]} x \bigg(  \overline{R}_{mobile}(s)- \xi \overline{R}_{static}(s) \bigg)$$
This proves the theorem.
\end{proof}
{\bf Remark:} The binary nature of the optimal policy in Theorem~\ref{theorem:policy-structure} makes is very easy to use the policy for optimal bandwidth allocation in a practical cellular network.
{\bf Comments on Fairness:}  Note that, each static user will 
asymptotically  receive positive throughput, since with positive probability 
a cell will have zero mobile user at a given time slot.  
On the other hand, a mobile user might get zero throughput in the current cell. 
In order to ensure a fair bandwidth sharing inside each cell, we describe in 
Section~\ref{section:fairness} how to share bandwidth between the two classes for a modified objective function which is motivated by the notion of $\alpha$-fairness (see \cite{chiang-fairness} for reference). The modified objective function 
ensures that both classes receive a positive throughput at the same time.

Let us denote the steady-state 
probability of occurrence of state $s$ by $g(s)$, with $\sum_s g(s)=1$. Under policy $\eta_{\xi}^*(\cdot)$, 
the optimal data rate for the mobile users per slot 
is given by: 
$$\overline{R}_{mobile}^*(\xi):=\sum_s g(s)  \overline{R}_{mobile}(s)   \eta_{\xi}^*(s).$$
Similarly, we define the optimal 
data rate of static users per slot by 
$$\overline{R}_{static}^*(\xi):=\sum_s g(s)  \overline{R}_{static}(s)  (1- \eta_{\xi}^*(s) ).$$

\begin{lemma}\label{lemma:mdp-optimal-rate-of-static-mobile-users-increasing-decreasing-in-xi}
 $\overline{R}_{mobile}^*(\xi)$ decreases with $\xi$, and  $\overline{R}_{static}^*(\xi)$ increases in $\xi$.
\end{lemma}
\begin{proof}
 See Appendix~\ref{appendix}.
\end{proof}

{\bf Error in estimating user location:} This issue is 
addressed in Section~\ref{subsection:error-in-location-estimation} in detail.

\section{Learning Algorithm for the Unconstrained Problem}
\label{section:learning-algorithm-given-lagrange-multiplier}
In  Section~\ref{section:mdp-formulation-for-the-unconstrained-problem}, we assumed that perfect knowledge of 
 $\overline{R}_{mobile}(s)$ and $\overline{R}_{static}(s)$ is available to the BS. 
 However, in practice, 
 unknown path-loss factor (since path-loss exponent and location of interfering base stations are unknown to the BS), 
 unknown shadowing variation over space 
 and unknown fading distribution will make it impossible for the base station to 
 compute $\overline{R}_{mobile}(s)$ and $\overline{R}_{static}(s)$. Hence, 
 the base station cannot use the simple policy structure given by Theorem~\ref{theorem:policy-structure}. 
However, the base station can get a feedback 
 from the users about how much data the 
 users were able to download between two successive decision instants; this can happen if the base station keeps on 
 sending data packets to the users, and the 
 users measure packet error rate in the received data and send feedback 
 to the base station before a new decision is made. In this section, we propose a sequential 
 bandwidth allocation and learning algorithm, which maintains a running 
 estimate of $\overline{R}_{mobile}(s)$ and $\overline{R}_{static}(s)$ for each state $s$, 
 and updates these running estimates as new user feedbacks 
 are gathered, so as to converge asymptotically to a stationary policy  solving  
 the unconstrained problem~(\ref{eqn:unconstrained_mdp_for_a_cell}).
 
  \begin{assumption}\label{assumption:fading-ergodic}
  The fading gain between any base station (serving or interfering) 
  and a specific location in the cell comes from an ergodic Markov process (across time slots)  
  taking values from a bounded subset of the nonnegative real line, 
  and it is identically distributed across locations in the cell and across various BSs.\qed
 \end{assumption}
 
 Note that, this assumption ensures that if we sample $R_{mobile}(s)$ infinitely often, 
 we can essentially average over fading, and 
 obtain a correct estimate of 
 $\overline{R}_{mobile}(s)$,  even though the slot duration $\sigma$ might be  smaller than the time required to 
average over all possible fading states by a mobile user. 
 
Note that, by Theorem~\ref{theorem:policy-structure}, we can restrict ourselves to the action space $\{0,1\}$ instead of $[0,1]$. 
With this reduced state space, we present our sequential bandwidth allocation and learning algorithm, which 
is motivated by the theory of stochastic approximation (see \cite{borkar08stochastic-approximation-book}).

\subsection{The Learning Algorithm}
{\em Some notation:}  Let $\eta_{\tau} \in \{0,1\}$ denote the decision to be taken 
 at decision instant $\tau$. Let $R_{mobile}(s)$ and $R_{static}(s)$ be the (random) 
 realization of the total rates received between decision 
 instant $\tau$ and decision instant $\tau+1$ by the mobile (resp., static) users, provided that 
 $\eta_{\tau}=1$ (resp., $\eta_{\tau}=0$).
 
 Fix any small number  
 $\epsilon>0$. Suppose that at the decision instant $\tau$, the Markov chain has reached state 
 $s$, and let the current estimates of $\overline{R}_{mobile}(s)$ and $\overline{R}_{static}(s)$ be $R_{mobile}^{(\tau)}(s)$ 
 and $R_{static}^{(\tau)}(s)$, respectively. 
 
 Let us define $\nu(s,1,\tau):=\sum_{t=1}^{\tau} \ind \{s_t=s,\eta_t=1\}$ and $\nu(s,0,\tau):=\sum_{t=1}^{\tau} \ind \{s_t=s,\eta_t=0\}$.

 Let $\{a(t)\}_{t \geq 1}$ be a decreasing sequence of positive numbers with $\sum_{t=1}^{\infty}a(t)=\infty$ 
 and $\sum_{t=1}^{\infty}a^2(t)<\infty$.

\begin{algorithm}\label{algorithm:learning-algorithm-single-timescale}
 Start with arbitrary $R_{mobile}^{(0)}(s)$ and $R_{static}^{(0)}(s)$. 
  
 {\em (Decision on bandwidth sharing:)} At decision instant $\tau$, 
 with probabilities $\frac{\epsilon}{2}$ each, allocate the entire bandwidth 
 to the static users (i.e., take $\eta_{\tau}=0$) or to the mobile users (i.e., take $\eta_{\tau}=1$). Else 
 (with probability $(1-\epsilon)$), allocate the entire bandwidth to mobile users (i.e., $\eta_{\tau}=1$) if 
 $R_{mobile}^{(\tau)}(s)-\xi R_{static}^{(\tau)}(s)>0$, allocate the entire bandwidth to static users (i.e., 
 $\eta_{\tau}=0$) if 
  $R_{mobile}^{(\tau)}(s)-\xi R_{static}^{(\tau)}(s)<0$, and allocate the entire bandwidth 
  arbitrarily either to SUs or to MUs if 
   $R_{mobile}^{(\tau)}(s)-\xi R_{static}^{(\tau)}(s)=0$.
   
 {\em (Updating/learning the estimates:)} 
 Just before the $(\tau+1)$-st decision instant, for each possible state $s$, make the following update:
 
 \begin{eqnarray*}
  R_{mobile}^{(\tau+1)}(s)&=&R_{mobile}^{(\tau)}(s)+a(\nu(s,1,\tau))\ind \{s(\tau)=s,\eta_{\tau}=1\}\\
  && \times \bigg(R_{mobile}(s)-R_{mobile}^{(\tau)}(s) \bigg)\\
    R_{static}^{(\tau+1)}(s)&=&R_{static}^{(\tau)}(s)+a(\nu(s,0,\tau))\ind \{s(\tau)=s,\eta_{\tau}=0\}\\
  && \times \bigg(R_{static}(s)-R_{static}^{(\tau)}(s) \bigg)\\
 \end{eqnarray*}  \qed
\end{algorithm}
\subsection{Optimality of the Learning Algorithm}
Let us denote the average expected reward per slot under Algorithm~\ref{algorithm:learning-algorithm-single-timescale} 
by $\lambda_{\epsilon}^*(\xi)$.

\begin{theorem}\label{theorem:single-timescale-convergence}
 Under Assumption~\ref{assumption:fading-ergodic} and Algorithm~\ref{algorithm:learning-algorithm-single-timescale}, for each state $s$, we have 
 $\lim_{\tau \rightarrow \infty} R_{mobile}^{(\tau)}(s)=\overline{R}_{mobile}(s)$ and 
 $\lim_{\tau \rightarrow \infty} R_{static}^{(\tau)}(s)=\overline{R}_{static}(s)$ 
 almost surely. Consequently, $\lim_{\epsilon \downarrow 0} \lambda_{\epsilon}^*(\xi)=\lambda^*(\xi)$ 
 (note that, $\epsilon$ cannot be taken to be equal to $0$).
\end{theorem}
\begin{proof}
 See Appendix~\ref{appendix}. 
\end{proof}

\subsection{Remarks}
\begin{itemize}
\item Theorem~\ref{theorem:single-timescale-convergence} tells us that in a practical cellular network where the shadowing realizations at all locations and the location of interfering base stations are not known, one can still learn  the asymptotically  optimal bandwidth sharing policy by learning only $\overline{R}_{mobile}(s)$ and $\overline{R}_{static}(s)$.
 \item At any state $s$, we randomize our decision with probabilities $\epsilon$ and $(1-\epsilon)$ 
 for the following reason. A sufficient 
 condition for the convergence of 
 $R_{mobile}^{(\tau)}(s)$ to $\overline{R}_{mobile}(s)$ and convergence of 
 $R_{static}^{(\tau)}(s)$ to $\overline{R}_{static}(s)$ is $\lim \inf_{\tau \rightarrow \infty}\frac{\nu(s,1,\tau)}{\tau} >0 $  
and $\lim \inf_{\tau \rightarrow \infty}\frac{\nu(s,0,\tau)}{\tau} >0 $ almost surely for each $s$; i.e., all state-action pairs should occur comparatively 
often. We ensure this by the proposed randomized decision making and using the fact that the states come from an ergodic discrete-time finite 
 state Markov chain. Very small or very large value of $\epsilon$ might lead to possibly sample-path dependent 
 slow convergence rate. 
 \item It is easy to see that: 
 $$|\lambda_{\epsilon}^*(\xi)-\lambda^*(\xi)| \leq \frac{\epsilon}{2 }\sum_{s}g(s)  \E  |R_{mobile}(s)-R_{static}(s) |.$$ 
 Hence, by choosing 
 $\epsilon$ small, we can achieve a mean reward per slot which is arbitrarily close to the optimal value, but the convergence 
 rate might be slow depending on the initial values of the iterates and the realization of the sample path.
 \item The above problem of yielding an average reward slightly different than $\lambda^*(\xi)$ can be solved 
 in the following way. At the decision instant $\tau$, instead of using the randomization with probability 
 $\epsilon$ (as defined in Algorithm~\ref{algorithm:learning-algorithm-single-timescale}), one could randomize 
 for state $s$ with a probability $\frac{\epsilon}{\nu(s,\tau)}$ where $\nu(s,\tau)$ is the number 
 of occurrence of state $s$ up to time $\tau$. Since the Markov chain is finite state, positive recurrent, irreducible 
and independent of the actions taken by the base 
 station, and since $\sum_{k=1}^{\infty}\frac{\epsilon}{k}=\infty$, by the second Borel-Cantelli lemma we can say 
 that $$\Pro (\lim_{\tau \rightarrow \infty}\nu(s,1,\tau)=\infty)=1;$$ this is sufficient to 
 prove Theorem~\ref{theorem:single-timescale-convergence}. However, we did not use this randomization probability 
 because it will not ensure the conditions 
 $\lim \inf_{\tau \rightarrow \infty}\frac{\nu(s,1,\tau)}{\tau} >0 $  
and $\lim \inf_{\tau \rightarrow \infty}\frac{\nu(s,0,\tau)}{\tau} >0 $ almost surely for each $s$, which is necessary 
for the convergence proof of the multi-timescale learning algorithm (Algorithm~\ref{algorithm:learning-algorithm-three timescale} in 
Section~\ref{subsection:learning-algorithm-three-timescale}) which is 
inspired by Algorithm~\ref{algorithm:learning-algorithm-single-timescale}.
 \item A special choice would be $a(t)=\frac{1}{t}$, which 
 will lead to sample averaging of the iterates (of course, with the imperfection 
 created by randomized sampling). But we use the general step size $a(t)$ 
 here because it will help in developing multi-timescale 
 learning algorithm for a constrained problem explained in 
 Section~\ref{section:learning-algorithm-constrained-problem}.
 \item The rate of convergence is dependent on sample path 
 (i.e., realization of arrival process and the fading process at various locations), and also 
 on the size of the state space. However, convergence is guaranteed by 
 Theorem~\ref{theorem:single-timescale-convergence} so long as the state 
 space is finite.
 \item Speed of convergence will also depend on the choice of $a(t)$; 
 however, choosing a suitable step size sequence is beyond the scope 
 of this paper and we propose to leave it for future research work in this domain.

\end{itemize}

\section{Learning Algorithm for the Constrained Problem}
\label{section:learning-algorithm-constrained-problem}

In Section~\ref{section:learning-algorithm-given-lagrange-multiplier}, we had provided a learning algorithm that 
solves problem~\eqref{eqn:unconstrained_mdp_for_a_cell} for a given $\xi$. However, let us recall from 
Theorem~\ref{theorem:how-to-choose-optimal-Lagrange-multiplier} that, in order to solve 
the constrained problem~\eqref{eqn:constrained_mdp_for_a_cell}, 
we need to choose an appropriate $\xi^*$. Since the transition structure 
of the MDP in Section~\ref{section:mdp-formulation-for-the-unconstrained-problem} might not be known apriori 
(as discussed in Section~\ref{section:learning-algorithm-given-lagrange-multiplier}), 
in this section we develop a sequential decision and learning algorithm for dynamic bandwidth sharing between 
the two classes of static and mobile users; this algorithm maintains an estimate of $\xi^*$ and updates this estimate 
each time user is observed before a new MU enters the cell. We prove asymptotic convergence of the policy to the set of 
optimal policies.

\subsection{Need for Randomization}\label{subsection:need-for-randomized-policy}
Note that, while an optimal Lagrange multiplier $\xi^*$ may exist for a feasible constraint $R_0$, 
the optimal policy $\eta_{\xi^*}^*(\cdot | \cdot)$ solving the constrained problem~(\ref{eqn:constrained_mdp_for_a_cell}) may not be a 
deterministic policy. This can be explained in the following way. By 
Lemma~\ref{lemma:mdp-optimal-rate-of-static-mobile-users-increasing-decreasing-in-xi}, the optimal 
per-slot sum data rate for static users $\overline{R}_{static}^*(\xi)$ 
increases with $\xi$. However, since there are finite number of states and 
only two actions $\{0,1\}$, there are finite number of deterministic policies in the class 
specified by Theorem~\ref{theorem:policy-structure}. The mapping from state space to action space can only 
change a finite number of times as we increase $\xi$ from $0$ to $\infty$, Hence, the plot of the optimal 
time-average sum rate of static users under policy $\eta_{\xi}^*(\cdot)$ (i.e., $\overline{R}_{static}^*(\xi)$), 
as a function of $\xi$, would 
look like an increasing staircase function where the discontinuities  correspond to the values of $\xi$ where, by increasing 
$\xi-$ to $\xi+$, the policy changes because the optimal action for exactly one state changes from $1$ to $0$. 
Let the set of $\xi$ values where this plot is discontinuous, be denoted by $\mathcal{S}$. 
Also, let $\mathcal{D}$ denote the set of values of mean data rate per slot for static users, which can be achieved 
only via $\eta_{\xi}^*(\cdot)$ by varying $\xi$ from $0$ to $\infty$. 

In light of the above discussion, it is clear that a way to meet the constraint in 
(\ref{eqn:constrained_mdp_for_a_cell}) with equality (if $R_0 \notin \mathcal{D}$) is to randomize between the two policies 
$\eta_{\xi^*+}^*(\cdot)$ and $\eta_{\xi^*-}^*(\cdot)$ at each decision instant, with probabilities 
$1-p$ and $p$ respectively; these two deterministic policies differ in the action for 
exactly one state (if $R_0 \notin \mathcal{D}$).

\subsection{A special randomization technique}
\label{subsection:a-special-randomization-technique}
In Algorithm~\ref{algorithm:learning-algorithm-three timescale} presented next, we 
implement this randomization in a slightly unconventional way in order to tackle certain technical issues. 
Let us recall the policy $\eta_{\xi}^*(\cdot)$ from Theorem~\ref{theorem:policy-structure}. 
We choose a very small number $\delta>0$ (choice of $\delta$ is explained in  
Algorithm~\ref{algorithm:learning-algorithm-three timescale} later in 
Section~\ref{subsection:learning-algorithm-three-timescale}), and define a probability density function 
$f_p(\cdot)$ (parametrized by  a probability $p$) as follows:

$f_p(y)=\frac{p}{\delta}$ if $y \in [-\delta,0]$, $f_p(y)=\frac{1-p}{\delta}$ if $y \in (0,\delta]$, and 
$f_p(y)=0$ for  all other values of $y$. 

For any given $\xi$, in each slot $\tau$ 
one can sample a random variable $\Delta_{\tau} \sim f_p$ ($\{\Delta_{\tau}\}_{\tau \geq 1}$ i.i.d. across $\tau$) and use 
the policy $\eta_{\xi+\Delta_{\tau}}^*(\cdot)$ (i.e., take action $\eta_{\xi+\Delta_{\tau}}^*(s(\tau))$ in slot $\tau$).  
If $\xi=\xi^*$ and $R_0$ does not belong to 
$\mathcal{D}$, then this scheme will correspond to 
randomizing between $\eta_{\xi^*+}^*(\cdot)$ and $\eta_{\xi^*-}^*(\cdot)$ with probabilities $1-p$ and $p$ in each slot 
(but this randomization is applicable to all possible values of $\xi$). 

Let the optimal value of $p$ for a 
given value of $\xi$ be denoted by $p^*(\xi)$; this is the optimal value of $p$ under multiplier $\xi$ so that 
the corresponding randomized algorithm (described just above using the probability density 
function $f_p(\cdot)$) meets the constraint with equality (if possible, 
given the value of $\xi$, as explained later in this section).

\begin{definition}\label{definition:the-set-mathcal-K}
The set $\mathcal{K}(R_0) \subset [0,1] \times [0,A]$ is defined to 
be the set of tuples $(p^*(\xi),\xi)$ under which the   
randomized policy described above meets the constraint in \eqref{eqn:constrained_mdp_for_a_cell} with equality. 
\end{definition}

\begin{assumption}\label{assumption:existence-of-optimal-xi}
 There exists $\xi^*>0$ and $p^*(\xi^*) \in [0,1]$ such that the corresponding randomized policy with these 
 parameters is optimal for the constrained problem~(\ref{eqn:constrained_mdp_for_a_cell}), while 
 the constraint is satisfied with equality. In other words, the set $\mathcal{K}(R_0)$ is nonempty.\qed
\end{assumption}

Note that, $\mathcal{K}(R_0)$ involves the function $p^*(\xi)$, and 
$ p^*(\xi)$ can be $0$ or $1$ also,  depending on the value of $\xi$. 
If $\xi$ is such that $\sum_s g(s) \Pro (\eta(s)=0|\xi,p) \overline{R}_{static}(s) <R_0  $ 
for all $p \in [0,1]$, then 
we will have $p^*(\xi)=0$. If $\xi$ is such that $\sum_s g(s) \Pro (\eta(s)=0|\xi,p) \overline{R}_{static}(s) >R_0  $ 
for all $p \in [0,1]$, then 
we will have $p^*(\xi)=1$. These two events happen {\em if} the value of $\xi$ does not fall within a $\delta$-neighbourhood 
of the element from $\mathcal{S}$ for which the constraint can be met 
with equality, and $R_0$ does not belong to $\mathcal{D}$; the constraint cannot be met 
with equality in this case under this $\xi$. If $R_0$ does not belong to $\mathcal{D}$ but the value of $\xi$ is within 
$\delta$-neighbourhood of the value from $\mathcal{S}$ which can achieve this $R_0$, then 
$p^*(\xi)$ can be anything in the interval $[0,1]$, depending on the value of $R_0$, so that the constraint 
is met with equality (if possible).

It is easy to prove the following:
\begin{lemma}\label{lemma:optimal-p-lipschitz-in-xi}
 $p^*(\xi)$ is Lipschitz continuous in $\xi$.
\end{lemma}
{\em Remark:} This lemma will be required to prove desired convergence of our learning 
Algorithm~\ref{algorithm:learning-algorithm-three timescale}. Note that, if we only randomize between policies 
$\eta_{\xi-\delta}^*(\cdot)$ and $\eta_{\xi+\delta}^*(\cdot)$ with probabilities $p^*(\xi)$ and 
$1-p^*(\xi)$ in each slot, then the result in this lemma will not hold. 
This is the specific reason that we consider this special form of 
randomization. 

\begin{definition}
Let the sets $\mathcal{S}$ and $\mathcal{D}$ change to $\mathcal{S}_{\epsilon}$ and $\mathcal{D}_{\epsilon}$ 
when, in each slot $\tau$, we decide $\eta_{\tau}=1$ or $\eta_{\tau}=0$ with probabilities $\frac{\epsilon}{2}$ each, and use 
the policy $\eta_{\xi}^*(\cdot)$ with probability $(1-\epsilon)$. Similarly, let the analogue 
of $\mathcal{K}(R_0)$ be $\mathcal{K}_{\epsilon}(R_0)$, and the analogue of $p^*(\xi)$ be $p_{\epsilon}^*(\xi)$. 
\end{definition}

\subsection{The Learning Algorithm Based on Two Timescale Stochastic Approximation}
\label{subsection:learning-algorithm-three-timescale}
Now we present a sequential bandwidth allocation and 
learning algorithm in order to solve the constrained problem~(\ref{eqn:constrained_mdp_for_a_cell}). The algorithm 
maintains running estimates   
$\{R_{mobile}^{(\tau)}(s),R_{static}^{(\tau)}(s)\}$ for all $s$, the Lagrange multiplier 
$\xi^{(\tau)}$, and the randomizing parameter $p^{(\tau)}$; 
this algorithm is motivated by two-timescale stochastic approximation 
(see \cite{borkar08stochastic-approximation-book}).

Suppose that at the decision instant $\tau$, the Markov chain has reached state 
 $s$, and let the current iterates be $R_{mobile}^{(\tau)}(s)$, $R_{static}^{(\tau)}(s)$, 
 $\xi^{(\tau)}$ and  $p^{(\tau)}$. 
Let us define $\mathcal{R}_{\tau}$ to be the collection of 
$\{R_{mobile}^{(\tau)}(s),R_{static}^{(\tau)}(s)\}$ for all $s$. We define $\eta_{\xi}^*(\cdot,\cdot)$ to be the same policy 
as $\eta_{\xi}^*(\cdot)$ given in Theorem~\ref{theorem:policy-structure}, except that 
$\overline{R}_{mobile}(s)$ and $\overline{R}_{static}(s)$ in Theorem~\ref{theorem:policy-structure} 
are replaced by the currents estimates $R_{mobile}^{(\tau)}(s)$ and $R_{static}^{(\tau)}(s)$ in slot 
$\tau$; the action taken in slot $\tau$ is $\eta_{\xi}^*(s(\tau), \mathcal{R}_{\tau})$.

 Let $\eta_{\tau} \in \{0,1\}$ denote the decision 
 at decision instant $\tau$. Let $R_{mobile}(s)$ and $R_{static}(s)$ be the (random) 
 realization of the total rates received between decision 
 instant $\tau$ and decision instant $\tau+1$ by the mobile (resp., static) users, provided that 
 $\eta_{\tau}=1$ (resp., $\eta_{\tau}=0$).

 Let us define $\nu(s,1,\tau):=\sum_{t=1}^{\tau} \ind \{s_t=s,\eta_t=1\}$ and $\nu(s,0,\tau):=\sum_{t=1}^{\tau} \ind \{s_t=s,\eta_t=0\}$. 
 
Let $\{a(t)\}_{t \geq 1}$ and $\{b(t)\}_{t \geq 1}$ 
  be decreasing sequences of positive numbers with 
 $\sum_{t=1}^{\infty}a(t)=\sum_{t=1}^{\infty}b(t)=\infty$, 
  $\sum_{t=1}^{\infty}a^2(t)<\infty$, $\sum_{t=1}^{\infty}b^2(t)<\infty$ and  
   $\lim_{t \rightarrow \infty}\frac{b(t)}{a(t)}=0$. 
  More specifically, we choose $a(t)=\frac{1}{t^{n_1}}$ and $b(t)=\frac{1}{t^{n_2}}$, 
  with $\frac{1}{2}<n_1<n_2 \leq 1$. Let $[x]_0^A$ denote the projection of $x$ on the compact 
  interval $[0,A]$, and let us choose the value of $A$ is chosen so large that $\overline{R}_{static}^*(\xi=A)>R_0$. 
  
 Fix any small number  
 $\epsilon>0$.  We choose $\delta>0$ to be a very small number, smaller than $1/10$-th of $\epsilon$ and 
 $1/10$-th of the smallest difference between two successive values of $\xi$ from the set 
 $\mathcal{S}_{\epsilon}$.

\begin{algorithm}\label{algorithm:learning-algorithm-three timescale}
 Start with $R_{mobile}^{(0)}(s)$, $R_{static}^{(0)}(s)$, $p^{(0)}$, $\xi^{(0)}$. 
   
 {\em (Decision on bandwidth sharing:)} At decision instant $\tau$, 
 with probabilities $\frac{\epsilon}{2}$ each, allocate the entire bandwidth 
 to the static users or to the mobile users. Else, (with probability $(1-\epsilon)$) sample a random variable 
 $\Delta_{\tau}$ (independent across $\tau$) 
 from the distribution $f_{p^{(\tau)}}(\cdot)$ independent of all other random variables, and use the policy 
 $\eta_{\xi^{(\tau)}+\Delta_{\tau}}^*(\cdot,\cdot)$ (i.e., take an action 
 $\eta_{\tau}=\eta_{\xi^{(\tau)}+\Delta_{\tau}}^*(s(\tau),\mathcal{R}_{\tau})$).  
 In other words, choose $\eta_{\tau}=1$ if 
 $R_{mobile}^{(\tau)}(s(\tau))-(\xi^{(\tau)}+\Delta_{\tau})R_{static}^{(\tau)}(s(\tau))>0$, choose 
 $\eta_{\tau}=0$ if $R_{mobile}^{(\tau)}(s(\tau))-(\xi^{(\tau)}+\Delta_{\tau})R_{static}^{(\tau)}(s(\tau))<0$ 
 and choose $\eta_{\tau}$ arbitrarily if $R_{mobile}^{(\tau)}(s(\tau))-(\xi^{(\tau)}+\Delta_{\tau})R_{static}^{(\tau)}(s(\tau))=0$.

 {\em (Updating/learning the estimates:)} 
  Just before the $(\tau+1)$-st decision instant, for each $s$, update as follows:
  
  \footnotesize
 \begin{eqnarray*}
  R_{mobile}^{(\tau+1)}(s)&=&R_{mobile}^{(\tau)}(s)+a(\nu(s,1,\tau))\ind \{s(\tau)=s,\eta_{\tau}=1\}\\
  && \times \bigg(R_{mobile}(s)-R_{mobile}^{(\tau)}(s) \bigg)\\
    R_{static}^{(\tau+1)}(s)&=&R_{static}^{(\tau)}(s)+a(\nu(s,0,\tau))\ind \{s(\tau)=s,\eta_{\tau}=0\}\\
  && \times \bigg(R_{static}(s)-R_{static}^{(\tau)}(s) \bigg)\\
   p^{(\tau+1)}&=&\bigg[ p^{(\tau)}+a(\tau) (\sum_s \ind \{s(\tau)=s,\eta_{\tau}=0\} \\
   && \times R_{static}(s)-R_0 ) \bigg]_0^1 \\
    \xi^{(\tau+1)}&=&\bigg[ \xi^{(\tau)}+b(\tau) (  R_0 - \\
   && \sum_s \ind \{s(\tau)=s,\eta_{\tau}=0\} R_{static}(s)  ) \bigg]_0^A \\
 \end{eqnarray*}
 \normalsize   \qed
\end{algorithm}
\normalsize

\subsection{Optimality of the Learning Algorithm for the Constrained Problem}
Let us denote the nonstationary, randomized policy induced by 
Algorithm~\ref{algorithm:learning-algorithm-three timescale} by $\eta^{(\epsilon)}(\cdot | \cdot,\cdot,\cdot,\cdot)$; the 
quantity $\eta^{(\epsilon)}(\cdot | s,\mathcal{R}_{\tau},\xi^{(\tau)},p^{(\tau)})$ denotes the probability distribution 
on the set of actions conditioned on the current state and the current values of the iterates.
\begin{theorem}\label{theorem:convergence-two-timescale-iteration}
 Under Assumption~\ref{assumption:fading-ergodic}, Assumption~\ref{assumption:existence-of-optimal-xi} and 
 Algorithm~\ref{algorithm:learning-algorithm-three timescale}, we have 
 $\lim_{\tau \rightarrow \infty} R_{mobile}^{(\tau)}(s)=\overline{R}_{mobile}(s)$ and  
 $\lim_{\tau \rightarrow \infty} R_{static}^{(\tau)}(s)=\overline{R}_{static}(s)$ for all $s$ almost surely. 
 Also, for any $\epsilon>0$, 
 $(p^{(\tau)},\xi^{(\tau)}) \rightarrow \mathcal{K}_{\epsilon}(R_0)$ almost surely as $\tau \rightarrow \infty$. 
  \end{theorem}
 \begin{proof}
  See Appendix~\ref{appendix}.
 \end{proof}
 
Let us denote 
 $$\overline{R}_{static}^{rand,\epsilon}:=\liminf_{N \rightarrow \infty}\frac{1}{N}\sum_{\tau=1}^{N} \E _{\eta^{(\epsilon)}(\cdot | \cdot,\cdot,\cdot,\cdot)} \bigg( (1-\eta_{\tau})\overline{R}_{static}(s(\tau)) \bigg)$$ 
 and 
  $$\overline{R}_{mobile}^{rand,\epsilon}:=\liminf_{N \rightarrow \infty}\frac{1}{N}\sum_{\tau=1}^{N} \E _{\eta^{(\epsilon)}(\cdot | \cdot,\cdot,\cdot,\cdot)} \bigg( \eta_{\tau} \overline{R}_{mobile}(s(\tau)) \bigg)$$

 \begin{corollary}\label{corollary:two-timescale-learning-optimal-for-constrained-problem}
  $\lim_{\epsilon \downarrow 0}  \overline{R}_{mobile}^{rand,\epsilon}$ and $\lim_{\epsilon \downarrow 0}  \overline{R}_{static}^{rand,\epsilon}$ 
  exist, and these limit values are equal to the optimal value of the objective in 
  the constrained problem~\eqref{eqn:constrained_mdp_for_a_cell} and $R_0$, respectively. 
 \end{corollary}
 \begin{proof}
  See Appendix~\ref{appendix}.
 \end{proof}

 {\em Remark:} Corollary~\ref{corollary:two-timescale-learning-optimal-for-constrained-problem} implies that 
 Algorithm~\ref{algorithm:learning-algorithm-three timescale} approximately solves the constrained 
 problem~\eqref{eqn:constrained_mdp_for_a_cell} for arbitrarily small $\epsilon>0$. This result, which is derived from Theorem~\ref{theorem:convergence-two-timescale-iteration}, allows us to optimally assign bandwidth between the static and mobile user classes even when the transition probability structure of the MDP is not known apriori.

\subsection{Remarks on Theorem~\ref{theorem:convergence-two-timescale-iteration}:}
\begin{itemize}
\item {\em Two timescales:} The update scheme is based on two timescale 
stochastic approximation (see \cite[Chapter~$6$]{borkar08stochastic-approximation-book}). 
Note that, $\lim_{t \rightarrow \infty}\frac{b(t)}{a(t)}=0$;  $\xi$ 
is adapted in the 
{\em slower} timescale, and $R_{mobile}$, $R_{static}$ and $p$ are updated in the {\em faster} 
timescale). The dynamics behaves as if the  
slower timescale update equation views the faster timescale iterates as quasi-static, 
while a faster timescale update 
equation views the slower timescale update equations as almost equilibrated; as if $\xi$ is being varied in a slow outer 
loop, while the other iterates are being varied in an inner loop. 

\item {\em Structure of the iteration:} Note that, the value of $\xi$ is increased whenever the sum data downloaded by  
static users between two successive decision instants is less than the target $R_0 $, so that 
more emphasis is given to the static user rate in the objective function. Under the same situation, 
the value of $p$ is reduced for the same reason. The goal is to converge to a randomized policy 
$\eta^{(\epsilon)}(\cdot | \cdot,\cdot,\cdot,\cdot)$ so that the corresponding randomized policy satisfies 
the constraint in \eqref{eqn:constrained_mdp_for_a_cell} with equality.

\item Algorithm~\ref{algorithm:learning-algorithm-three timescale}  
induces a nonstationary  
policy. But, by 
Theorem~\ref{theorem:convergence-two-timescale-iteration} and 
Corollary~\ref{corollary:two-timescale-learning-optimal-for-constrained-problem}, the sequence of policies generated by 
Algorithm~\ref{algorithm:learning-algorithm-three timescale} converges close  to the set 
of optimal stationary, randomized policies for the constrained problem (\ref{eqn:constrained_mdp_for_a_cell}).
\end{itemize}

\section{Fair bandwidth sharing between static and mobile user classes}\label{section:fairness}
In previous sections, the proposed dynamic bandwidth sharing schemes do not guarantee nonzero throughput to each user all the time. 
While such schemes are suitable for elastic traffic applications, they are not at all suitable for streaming 
applications  such as online video watching or voice call. In fact, opportunistic bandwidth sharing depending 
on user location as described before will result in unfair sharing of bandwidth. In order to incorporate fairness 
constraint into the opportunistic bandwidth sharing problem, we modify the objective function presented in 
Section~\ref{subsection:smdp-formulation}. 

Let us denote $\overline{R}_{static}^{\alpha}(s):=\E R_{static}^{\alpha}(s)$ and 
$\overline{R}_{mobile}^{\alpha}(s):=\E R_{mobile}^{\alpha}(s)$ where $\alpha$ is a real number.

In this section, we consider the following unconstrained problem: 

\footnotesize
\begin{eqnarray}
 \sup_{\eta(\cdot | \cdot)} \liminf_{N \rightarrow \infty}\frac{1}{N}\sum_{\tau=1}^{N} \E _{\eta(\cdot | \cdot)} \bigg( \eta_{\tau}^{\alpha}\overline{R}_{mobile}^{\alpha}(s(\tau)) \nonumber\\
+\xi (1-\eta_{\tau})^{\alpha}\overline{R}_{static}^{\alpha}(s(\tau)) \bigg) \label{eqn:unconstrained_mdp_for_a_cell_alpha_fairness}
\end{eqnarray}
\normalsize
and also the associated constrained problem as follows:

\footnotesize
\begin{eqnarray}
 &&\sup_{\eta(\cdot | \cdot)} \liminf_{N \rightarrow \infty}\frac{1}{N}\sum_{\tau=1}^{N} \E _{\eta(\cdot | \cdot)} \bigg( \eta_{\tau}^{\alpha} \overline{R}_{mobile}^{\alpha}(s(\tau)) \bigg) \nonumber\\
 &s.t.,&  \liminf_{N \rightarrow \infty}\frac{1}{N}\sum_{\tau=1}^{N} \E _{\eta(\cdot | \cdot)} \bigg( (1-\eta_{\tau})^{\alpha}\overline{R}_{static}^{\alpha}(s(\tau)) \bigg) \geq R_0  \nonumber\\
 \label{eqn:constrained_mdp_for_a_cell_alpha_fairness}
\end{eqnarray}
\normalsize
This objective function is motivated by the notion of $\alpha$-fairness (see \cite{chiang-fairness}). The intuition is that the degree of fairness in resource allocation between mobile and static user classes can be controlled by appropriately tuning  $\alpha$ in \eqref{eqn:unconstrained_mdp_for_a_cell_alpha_fairness} and \eqref{eqn:constrained_mdp_for_a_cell_alpha_fairness}.

Let us recall the proof of Theorem~\ref{theorem:policy-structure}; Theorem~\ref{theorem:policy-structure} 
provides optimal allocation for the special case $\alpha=1$. In general, the function 
$ x^{\alpha}\overline{R}_{mobile}^{\alpha}(s))+\xi (1-x)^{\alpha} \overline{R}_{static}^{\alpha}(s) $ is strictly convex in $x$ 
for $\alpha>1$ or for $\alpha<0$, strictly concave in $x$ for $\alpha \in (0,1)$, independent of $x$ for $\alpha=0$, and linear in $x$ for 
$\alpha=1$. 
Hence, for $\alpha>1$ or for $\alpha<0$, the 
optimization 
$\max_{\alpha \in [0,1]}  x^{\alpha}\overline{R}_{mobile}^{\alpha}(s))+\xi (1-x)^{\alpha} \overline{R}_{static}^{\alpha}(s) $ will 
always have $x^*=\eta^*(s) \in \{0,1\}$. On the other hand, for $\alpha \in (0,1)$, the optimal 
value of $x$ lies in $(0,1)$. Hence, in this section, we focus only on $\alpha \in (0,1)$. Within this interval, $\alpha$ close to $0$ 
yields more egalitarian solution, whereas $\alpha$ close to $1$ provides more opportunistic bandwidth sharing.

Note that, $\alpha$ in this current paper is slightly different from $\alpha$ in  \cite{chiang-fairness}).\footnote{In \cite{chiang-fairness}, if 
the resource (e.g., rate) allocated to user $i$ is $r_i$, then $\alpha$-fair utility function is given by $\sum_i \frac{r_i^{1-\alpha}-1}{1-\alpha}$. 
For $\alpha>1$, it requires minimization of $\sum_i r_i^{1-\alpha}$, and for $\alpha<1$, it requires maximization of $\sum_i r_i^{1-\alpha}$. For 
$\alpha=1$, this reduces to maximization of $\sum_i \log (r_i)$, which is called proportional fair allocation; we exclude this case 
because this, in our current problem, will result in bandwidth sharing independent of the value of $s$. In our current paper, we 
have used $\alpha$ in a slightly different sense than \cite{chiang-fairness}, though the broad concept of fairness is the same as 
\cite{chiang-fairness}; the only difference is that, unlike \cite{chiang-fairness}, we consider fairness only between two classes of users, and any 
bandwidth sharing policy can be employed within a single class.}

\subsection{Policy structure under perfect knowledge of $\overline{R}_{static}^{\alpha}(s)$ and 
$\overline{R}_{mobile}^{\alpha}(s)$}\label{subsection:perfect-knowledge-smdp-modified}
Let us assume the availability of perfect knowledge at the base station (as assumed in 
Section~\ref{section:mdp-formulation-for-the-unconstrained-problem}), but let us consider the 
objective function \eqref{eqn:unconstrained_mdp_for_a_cell_alpha_fairness}. Similar to Theorem~\ref{theorem:policy-structure}, 
the optimal action at state $s$ is given by $$\eta_{\xi}^*(s)=\arg\max_{x \in [0,1]}   \bigg( x^{\alpha}\overline{R}_{mobile}^{\alpha}(s))+\xi (1-x)^{\alpha} \overline{R}_{static}^{\alpha}(s) \bigg).$$  
Differentiation w.r.t. $x$ and setting the derivative equal to $0$, we obtain:

\begin{equation}
 \eta_{\xi}^*(s)=\frac{(\xi \overline{R}_{static}^{\alpha}(s))^{\frac{1}{\alpha-1}}}{ (\xi \overline{R}_{static}^{\alpha}(s))^{\frac{1}{\alpha-1}}+ (\overline{R}_{mobile}^{\alpha}(s))^{\frac{1}{\alpha-1}} }
 \label{eqn:policy-structure-general-alpha}
\end{equation}
The optimal policy is to allocate $\eta_{\xi}^*(s)$ fraction of bandwidth to mobile  users and $1-\eta_{\xi}^*(s)$ 
fraction of bandwidth to static users whenever the system reaches state $s$. Let this optimal policy be denoted by 
$\eta_{\xi}^*(\cdot)$.

The following lemma is easy to prove:
\begin{lemma}\label{lemma:optimal-eta-decreasing-continuous-in-xi}
 $\eta_{\xi}^*(s)$ in \eqref{eqn:policy-structure-general-alpha} is strictly decreasing and Lipschitz 
 continuous in $\xi$ for all $s$ and for all  $\xi>0$.
\end{lemma}

\subsection{Learning Algorithm for the Unconstrained Problem~\eqref{eqn:unconstrained_mdp_for_a_cell_alpha_fairness}}

Let us now consider imperfect knowledge at the base station as assumed in 
Section~\ref{section:learning-algorithm-given-lagrange-multiplier}, but with the modified objective function 
\eqref{eqn:unconstrained_mdp_for_a_cell_alpha_fairness} with $\alpha \in (0,1)$. We seek to propose learning algorithms 
as done in Algorithm~\ref{algorithm:learning-algorithm-single-timescale}.

Note that, since a strictly positive fraction 
of bandwidth is always allocated to static and mobile users at any given time,  samples of $R_{mobile}^{\alpha}(s)$ and 
$R_{static}^{\alpha}(s)$ are always available whenever the system reaches state $s$. As a result of this and the 
positive recurrence of the Markov chain associated with state evolution, the estimates of $\overline{R}_{mobile}^{\alpha}(s)$ and 
$\overline{R}_{static}^{\alpha}(s)$ will be updated infinitely often for each state $s$, and there is no need to do the randomization 
with probability $\epsilon$ as described in Section~\ref{section:learning-algorithm-given-lagrange-multiplier}.

Let $\eta_{\tau} \in [0,1]$ denote the decision 
 at decision instant $\tau$. Let $R_{mobile}^{\alpha}(s)$ and $R_{static}^{\alpha}(s)$ denote the samples (of the $\alpha$-th moment 
 of the corresponding rates) 
 obtained between decision  instant $\tau$ and decision instant $\tau+1$.
 
 Suppose that at the decision instant $\tau$, the Markov chain has reached state 
 $s$, and let the current estimates of $\overline{R}_{mobile}^{\alpha}(s)$ and $\overline{R}_{static}^{\alpha}(s)$ be 
 $R_{mobile,\alpha}^{(\tau)}(s)$ and $R_{static,\alpha}^{(\tau)}(s)$, respectively. 
 
 Let  $\nu(s,\tau):=\sum_{t=1}^{\tau} \ind \{s_t=s\}$. Let $\{a(t)\}_{t \geq 1}$ be 
 a decreasing sequence of positive numbers with $\sum_{t=1}^{\infty}a(t)=\infty$ 
 and $\sum_{t=1}^{\infty}a^2(t)<\infty$. 

We propose the following algorithm to learn the 
optimal policy for problem~\eqref{eqn:unconstrained_mdp_for_a_cell_alpha_fairness}.

\begin{algorithm}\label{algorithm:learning-algorithm-single-timescale-alpha-fairness}
 Start with any arbitrary initial estimates $R_{mobile,\alpha}^{(0)}(s)>0$ and $R_{static,\alpha}^{(0)}(s)>0$. 
 
  At decision instant $\tau$, if the system is 
 at state $s$, allocate the following fraction of bandwidth to the mobile users: 
 \begin{equation}
 \eta_{\tau}=\frac{(\xi R_{static,\alpha}^{(\tau)}(s))^{\frac{1}{\alpha-1}}}{ (\xi R_{static,\alpha}^{(\tau)}(s))^{\frac{1}{\alpha-1}}+ (R_{mobile,\alpha}^{(\tau)}(s))^{\frac{1}{\alpha-1}} }
 \label{eqn:policy-structure-single-timescale-learning-general-alpha}
\end{equation}
  
 Just before the $(\tau+1)$-st decision instant, for each possible state $s$, make the following update:
  \begin{eqnarray*}
  R_{mobile,\alpha}^{(\tau+1)}(s)&=&R_{mobile,\alpha}^{(\tau)}(s)+a(\nu(s,\tau))\ind \{s(\tau)=s\}\\
  && \times \bigg(R_{mobile}^{\alpha}(s)-R_{mobile,\alpha}^{(\tau)}(s) \bigg)\\
    R_{static,\alpha}^{(\tau+1)}(s)&=&R_{static,\alpha}^{(\tau)}(s)+a(\nu(s,\tau))\ind \{s(\tau)=s \}\\
  && \times \bigg(R_{static}^{\alpha}(s)-R_{static,\alpha}^{(\tau)}(s) \bigg)\\
 \end{eqnarray*}
 \qed
\end{algorithm}
\begin{theorem}\label{theorem:single-timescale-convergence-alpha-fairness}
 Under Assumption~\ref{assumption:fading-ergodic} and 
 Algorithm~\ref{algorithm:learning-algorithm-single-timescale-alpha-fairness}, for each state $s$, we have 
 $\lim_{\tau \rightarrow \infty} R_{mobile,\alpha}^{(\tau)}(s)=\overline{R}_{mobile}^{\alpha}(s)$ and 
 $\lim_{\tau \rightarrow \infty} R_{static,\alpha}^{(\tau)}(s)=\overline{R}_{static}^{\alpha}(s)$ 
 almost surely. 
\end{theorem}
\begin{proof}
 The proof is similar to the proof of Theorem~\ref{theorem:single-timescale-convergence}.
\end{proof}

\subsection{Learning Algorithm for the constrained problem~\eqref{eqn:constrained_mdp_for_a_cell_alpha_fairness}}
\label{subsection:learning-two-timescale-alpha-fairness}
In this subsection, we seek to propose learning algorithms for the constrained 
problem~\eqref{eqn:constrained_mdp_for_a_cell_alpha_fairness}, in a way similar to 
Section~\ref{subsection:learning-algorithm-three-timescale}. Let us define 
$$\overline{R}_{mobile,\alpha}^*(\xi):=\sum_s g(s) \overline{R}_{mobile}^{\alpha}(s)   (\eta_{\xi}^*(s))^{\alpha}$$ and 
$$\overline{R}_{static,\alpha}^*(\xi):=\sum_s g(s) \overline{R}^{\alpha}_{static}(s)   (1-\eta_{\xi}^*(s))^{\alpha},$$  
where $\eta_{\xi}^*(s)$ is defined in \eqref{eqn:policy-structure-general-alpha}. 
By Lemma~\ref{lemma:optimal-eta-decreasing-continuous-in-xi}, $\overline{R}_{static,\alpha}^*(\xi)$ is strictly 
increasing and continuous in $\xi$. Hence, if the constraint in 
\eqref{eqn:constrained_mdp_for_a_cell_alpha_fairness} is feasible, then there exists one $\xi^*>0$ such that the 
constraint is met with equality under the optimal policy given in Section~\ref{subsection:perfect-knowledge-smdp-modified} 
with $\xi=\xi^*$, i.e., $\overline{R}_{static,\alpha}^*(\xi^*)=R_0$ under $\eta_{\xi^*}^*(\cdot)$. 

Now we propose a sequential bandwidth allocation and learning algorithm (based on 
{\em single timescale} stochastic approximation) that will solve 
problem~\eqref{eqn:constrained_mdp_for_a_cell_alpha_fairness}.

Suppose that at the decision instant $\tau$, the Markov chain has reached state 
 $s$, and let the current estimates of $\overline{R}_{mobile}^{\alpha}(s)$, $\overline{R}_{static}^{\alpha}(s)$ and $\xi^*$ be 
 $R_{mobile,\alpha}^{(\tau)}(s)$, $R_{static,\alpha}^{(\tau)}(s)$ and $\xi^{(\tau)}$, respectively. 
 Let  $\nu(s,\tau):=\sum_{t=1}^{\tau} \ind \{s_t=s\}$.
 
 Let $\eta_{\tau} \in [0,1]$ denote the decision 
 at decision instant $\tau$. Let $R_{mobile}^{\alpha}(s)$ and $R_{static}^{\alpha}(s)$ denote the samples 
 obtained between decision  instant $\tau$ and decision instant $\tau+1$.

Let $\{a(t)\}_{t \geq 1}$ be a decreasing sequence of positive numbers with $\sum_{t=1}^{\infty}a(t)=\infty$ 
 and $\sum_{t=1}^{\infty}a^2(t)<\infty$. The numbers $B>0$ and $A>B$ are such that $\xi^* \in (B,A)$.
 
\begin{algorithm}\label{algorithm:learning-algorithm-constrained-problem-alpha-fairness}
 Start with any arbitrary initial estimates $R_{mobile,\alpha}^{(0)}(s)>0$, $R_{static,\alpha}^{(0)}(s)>0$ and $\xi^{(0)}$. \
 
  At decision instant $\tau$, if the system is 
 at state $s$, allocate the following fraction of bandwidth to the mobile users: 
 \begin{equation}
 \eta_{\tau}=\frac{(\xi^{(\tau)} R_{static,\alpha}^{(\tau)}(s))^{\frac{1}{\alpha-1}}}{ (\xi^{(\tau)} R_{static,\alpha}^{(\tau)}(s))^{\frac{1}{\alpha-1}}+ (R_{mobile,\alpha}^{(\tau)}(s))^{\frac{1}{\alpha-1}} }
 \label{eqn:policy-structure-constrained-problem-learning-general-alpha}
\end{equation}

Just before the $(\tau+1)$-st decision instant, for each possible state $s$, make the following update:
 
 \footnotesize
 \begin{eqnarray*}
  R_{mobile,\alpha}^{(\tau+1)}(s)&=&R_{mobile,\alpha}^{(\tau)}(s)+a(\nu(s,\tau))\ind \{s(\tau)=s\}\\
  && \times \bigg(R_{mobile}^{\alpha}(s)-R_{mobile,\alpha}^{(\tau)}(s) \bigg)\\
    R_{static,\alpha}^{(\tau+1)}(s)&=&R_{static,\alpha}^{(\tau)}(s)+a(\nu(s,\tau))\ind \{s(\tau)=s \}\\
  && \times \bigg(R_{static}^{\alpha}(s)-R_{static,\alpha}^{(\tau)}(s) \bigg)\\
  \xi^{(\tau+1)}&=&\bigg[ \xi^{(\tau)}+a(\tau) (  R_0 -  \sum_s \ind \{s(\tau)=s\} R_{static}^{\alpha}(s)  ) \bigg]_B^A \\
 \end{eqnarray*}
 \normalsize
 \qed
\end{algorithm}
\normalsize
\begin{theorem}\label{theorem:single-constrained-problem-convergence-alpha-fairness}
 Under Assumption~\ref{assumption:fading-ergodic} and 
 Algorithm~\ref{algorithm:learning-algorithm-constrained-problem-alpha-fairness}, we have 
 $\lim_{\tau \rightarrow \infty} R_{mobile,\alpha}^{(\tau)}(s)=\overline{R}_{mobile}^{\alpha}(s)$ and 
 $\lim_{\tau \rightarrow \infty} R_{static,\alpha}^{(\tau)}(s)=\overline{R}_{static}^{\alpha}(s)$ for all $s$, and 
 $\xi^{(\tau)} \rightarrow \xi^*$ (if there exists $\xi^* \geq 0$ such that $\overline{R}_{static,\alpha}^*(\xi^*)=R_0$) 
 almost surely. 
\end{theorem}
\begin{proof}
 The proof is similar to the proof of Theorem~\ref{theorem:convergence-two-timescale-iteration}.
\end{proof}

{\em Remark:} If $\xi^{(\tau)}=0$ for some $\tau$, the entire bandwidth is allocated to the class of mobile users at that 
decision instant. To avoid this, we always maintain $\xi^{(\tau)} \geq B>0$.

\section{Performance improvement through opportunistic bandwidth allocation: a numerical study}
\label{section:numerical-work-performance-improvement}
In this section, we numerically explore the improvement in performance for static and mobile users via opportunistic bandwidth allocation. 

\subsection{Asymptotic performance Improvement for various combinations of $\alpha$ and $\theta$}
\label{subsection:performance-improvement-vs-alpha}
We consider the following simulation environment:
\begin{itemize}
 \item The base stations are located on the corners of a regular grid; the set of locations of base stations is 
 given by $\{(1000i,1000j): -10 \leq i \leq 10, -10 \leq j \leq 10)\}$, where the unit of distance in the 
 $xy$~plane is meter. Hence, the smallest 
 distance between two base stations is $1000$~m. We consider Voronoi cells under this realization of the base stations. 
 The base station whose cell is under consideration is located at the origin, and its cell 
 is a $1000$~m $\times$ $1000$~m square with the origin at its center. 
 \item Path loss at a distance $r$ is $r^{-\beta}$ with the path-loss exponent $\beta=4$. 
 There is no shadowing and fading in the 
 wireless propagation environment. However, later we will also demonstrate the convergence rate of Algorithm~\ref{algorithm:learning-algorithm-three timescale} in presence of shadowing and fading.
 \item All base stations are transmitting at the same power levels. 
 Since we do not assume any thermal noise at the receiving nodes, and sine the 
 signal-to-interference-ratio (SIR) remains unchanged if the transmit power 
 of each base station is multiplied by the same factor, we can safely assume that 
 the transmit power of each base station is $1$~unit.
 \item There are $500$~static users inside the cell containing the origin, and their locations are chosen independently with uniform distribution from the cell. 
 \item Two roads along the $x=25$ line and $y=50$ line intersect the cell under consideration.  Each of these $1000$~m long line segments are divided into $10$ segments of length $100$~m each (it can be segmented further to the shadowing decorrelation distance level). We assume that mobile users 
 enter the cell along these lines at a velocity $50$~m/sec, and the slot duration is $2$~sec so that each MU covers one $100$~m distance 
 segment in one slot, 
 i.e., each MU traverses the cell in $20$~seconds ($10$~slots). 
 \item The number of arrivals (of MUs) to the cell in each slot is $50$ times a  Bernoulli distributed random variable with mean $\frac{1}{\theta}$; this is a batch arrival process. 
 \item We assume that, if the entire  bandwidth 
 (assumed to be $1$~unit) is allocated to a single MU, then the total amount of data this MU can download over a $100$~m long line segment 
 is given by $\log_2(1+SIR_{centre})$ where $SIR_{centre}$ is the SIR value at the center of the segment. For example, the total data rate 
 assigned to a MU when it is crossing the distance between $(-500,50)$ and $(-400,50)$ in a single slot is given by 
 $\log_2(1+SIR_{(-450,50)})$ (provided that the entire bandwidth is allocated to this user). On the other hand, the amount of data downloaded by a static 
 user when the entire bandwidth is allocated to this user is given by $\log_2(1+SIR)$ when the SIR corresponds to the location of the 
 static user.
\end{itemize}

\begin{table}
\footnotesize
\centering
\begin{tabular}{|c |c |c |c |c |c |c|}
\hline
 $\alpha$ &   $\theta$ & $\xi$ & $\overline{R}_{mobile,\alpha}^{equal}$  & $\overline{R}_{static,\alpha}^{equal}$ & $\overline{R}_{mobile,\alpha}^*$ & $\overline{R}_{static,\alpha}^*$ \\ \hline 
 0.1  &    0.1    &    2.3     &    0.4867      &    0.4655     &     0.5263      &         0.4654          \\ \hline
  0.1   &   0.9   &     1     &         0.7675    &     0.6161    &   0.7680    &    0.6212          \\ \hline
  0.9    &   0.1       &     1.5     &   0.0237      &      0.2539     &   0.0387       &   0.2780           \\\hline
 0.9     &   0.9      &  1.9    &     0.0935     &   0.0275   &     0.0941     &      0.0289             \\ \hline
 \end{tabular}
 \caption{Comparison of equal bandwidth sharing among all users against opportunistic (dynamic) bandwidth sharing between the static 
 and mobile user classes, for various combinations of $\alpha$ and $\theta$ (and correspondingly appropriate choice of $\xi$). 
 Under dynamic bandwidth sharing (columns~$6$ and $7$ in the table), 
 it is assumed that all users in the same user class (static or mobile) share equally 
 the bandwidth available to that class at any moment. The notation has been defined in the text.}
\normalsize
\label{table:comparison-equal-allocation-vs-opportunistic}
\end{table}

Since the stochastic approximation algorithms presented in this paper asymptotically converge 
to the optimal value, we first consider perfect knowledge scenario where the MDP transition and cost structures are 
known to the decision maker. 
For opportunistic (i.e., location-dependent) bandwidth allocation, we assume that all 
users in the same class (i.e., static or mobile) share equal bandwidth among themselves  all the time.

%

\begin{figure*}[t]
\includegraphics[width=0.33 \linewidth, height=5cm]{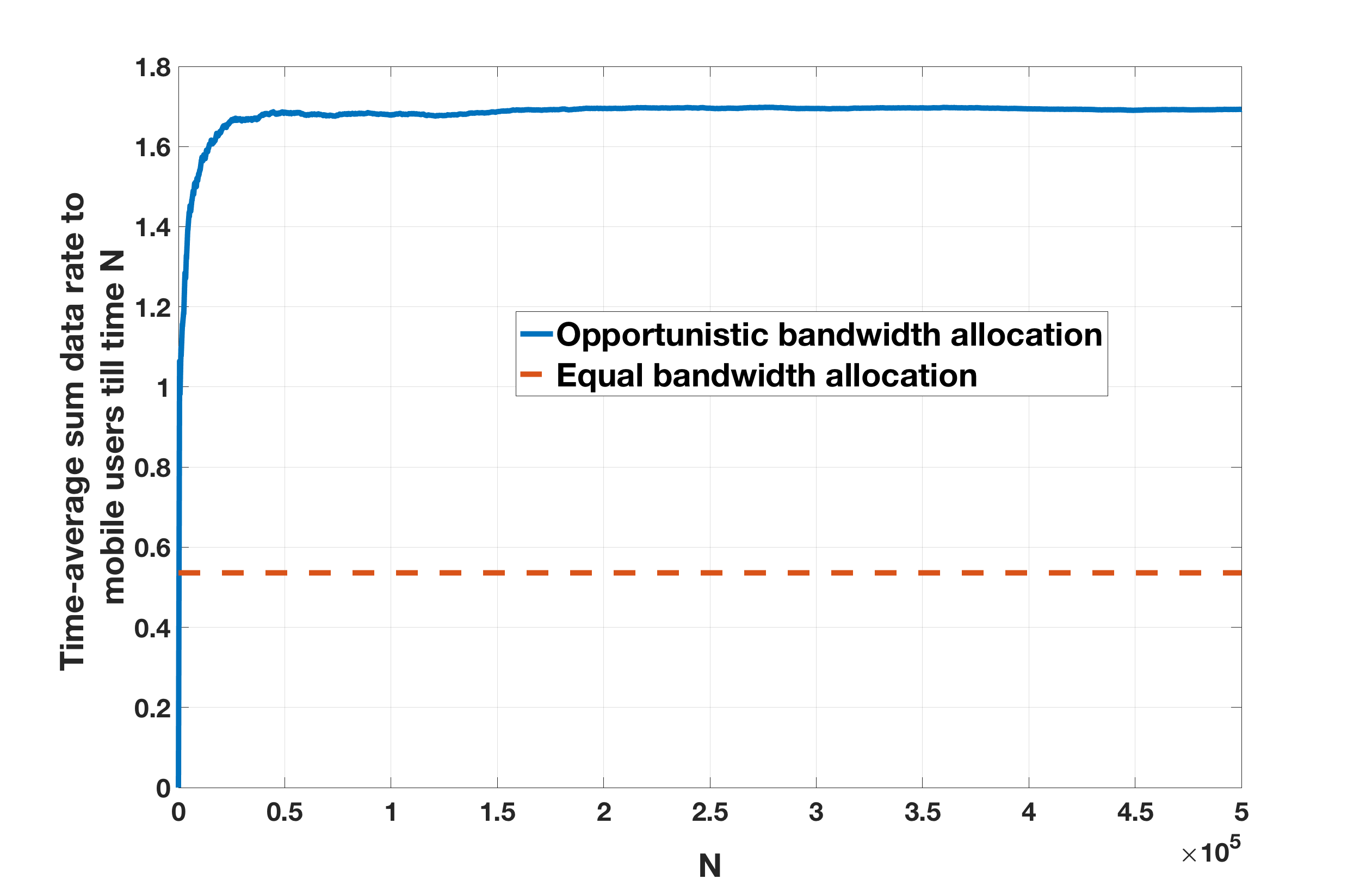}
\includegraphics[width=0.33 \linewidth, height=5cm]{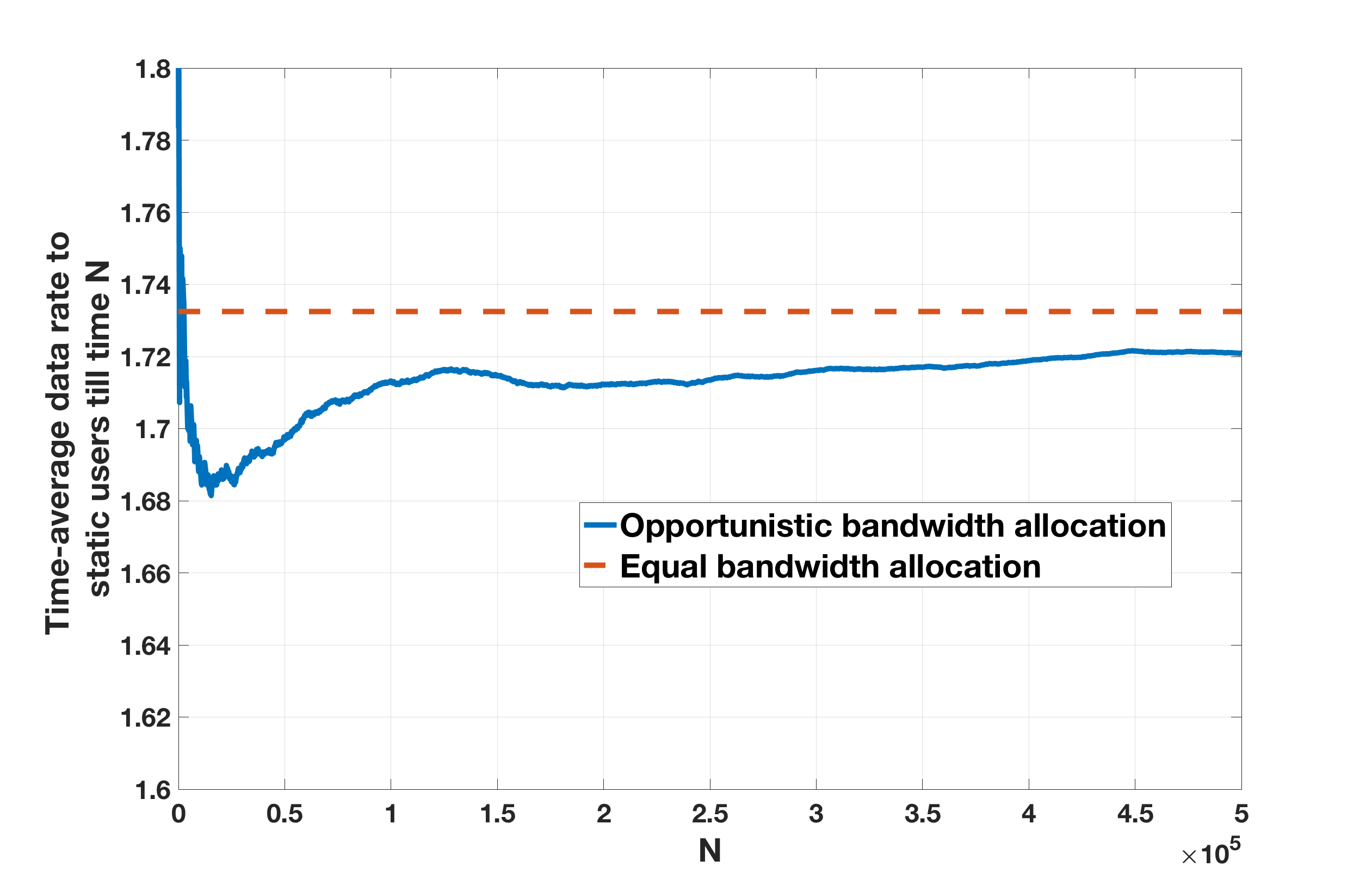}
\includegraphics[width=0.33 \linewidth, height=5cm]{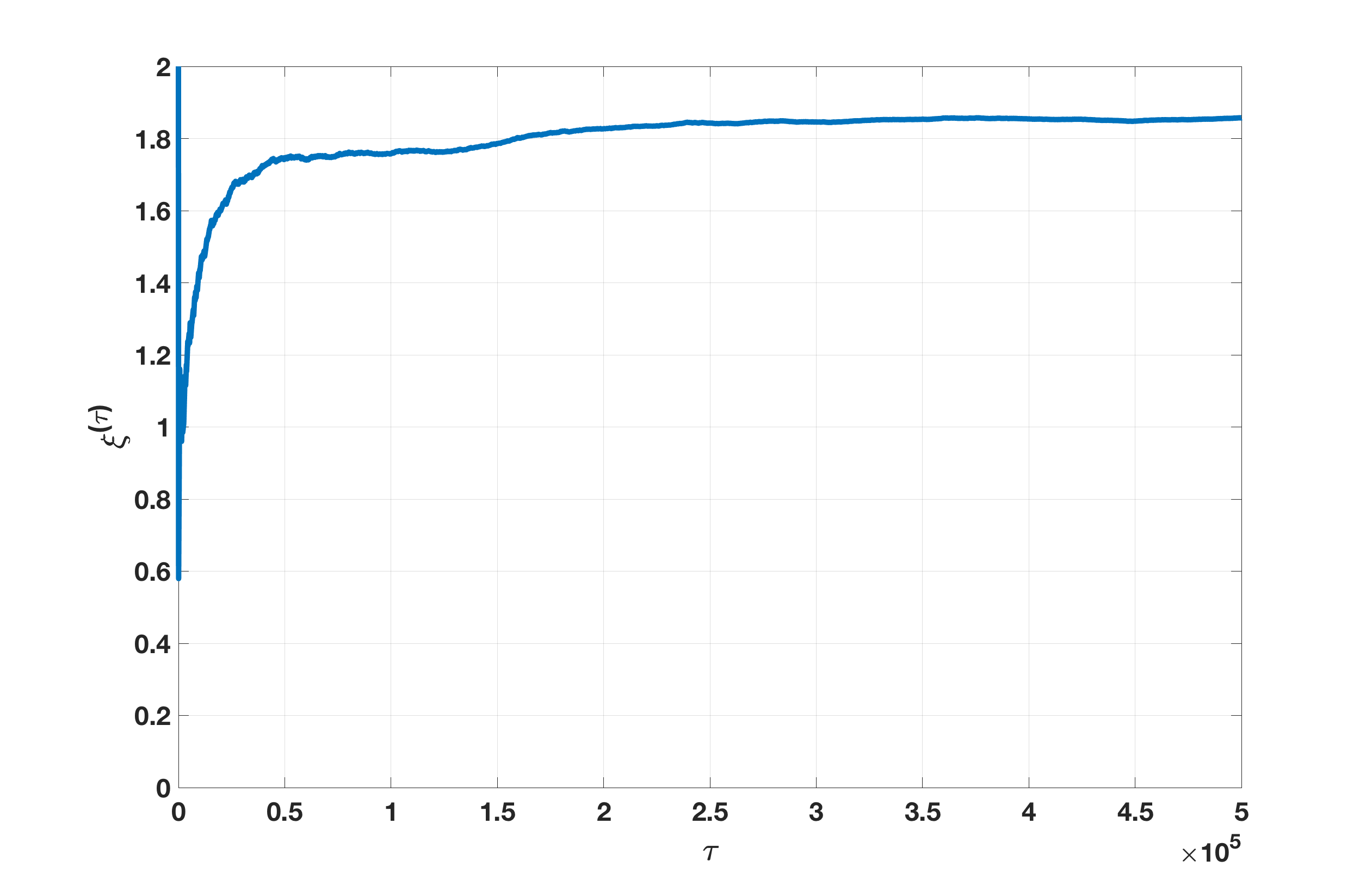}
\caption{Convergence of Algorithm~\ref{algorithm:learning-algorithm-three timescale}.}
\label{fig:convergence-plots}
\vspace{0mm}
\end{figure*}

We have done extensive simulation over a range of parameter values, for various realizations of the location of static users. In this section, we only 
provide a few of them to illustrate the performance gains and trade-offs. 

We first focus on the problem 
\eqref{eqn:constrained_mdp_for_a_cell_alpha_fairness} for $\alpha \in (0,1)$ for comparison. 
For each combination of $\alpha$ and $\theta$, we first 
compute non-opportunistic performance metrics 
$\overline{R}_{mobile,\alpha}^{equal}$ and $\overline{R}_{static,\alpha}^{equal}$ which are analogous to 
$\overline{R}_{mobile,\alpha}^*$ and $\overline{R}_{static,\alpha}^*$ defined in 
Section~\ref{subsection:learning-two-timescale-alpha-fairness} (with $\xi$ dropped from the notation), 
except that $\overline{R}_{mobile,\alpha}^{equal}$ and $\overline{R}_{static,\alpha}^{equal}$ 
are calculated assuming equal bandwidth sharing among all static and mobile users at any point of time. Then 
we chose an appropriate value of $\xi$ (for a given $\alpha$ and $\theta$) so that, under the corresponding optimal 
policies given in  
Section~\ref{subsection:perfect-knowledge-smdp-modified} with this choice of $\xi$, the constraint in 
  \eqref{eqn:constrained_mdp_for_a_cell_alpha_fairness} is (approximately) met with equality; clearly, 
our objective is to solve the constrained problem  \eqref{eqn:constrained_mdp_for_a_cell_alpha_fairness}. 
The quantities $\overline{R}_{mobile,\alpha}^*$ and $\overline{R}_{static,\alpha}^*$ under $\alpha=1$ become  
$\overline{R}_{mobile}^*$ and $\overline{R}_{static}^*$ (defined in Section~\ref{subsection:policy-structure-and-computation}). 
Our goal is to see how much improvement is possible (via opportunistic bandwidth sharing) in the time-average sum data rate 
of mobile users which keeping the same quantity unchanged 
for static users. The results are summarized in Table~\ref{table:comparison-equal-allocation-vs-opportunistic}. Note that, each row in Table~\ref{table:comparison-equal-allocation-vs-opportunistic} corresponds to an independent set of static user locations.

From Table~\ref{table:comparison-equal-allocation-vs-opportunistic}, we observe that even $60 \%$  
improvement is possible in the 
time-average throughput of mobile users, while keeping the time-average 
throughput of static users almost unchanged; this clearly shows that it is worth employing 
the proposed opportunistic bandwidth allocation 
algorithms in cellular networks. We also 
observe that the margin of performance improvement decreases as $\alpha$ becomes smaller. This happens 
because of two reasons: (i) choice of $\alpha \in (0,1)$ allows more fair allocation at the cost of opportunistic gain, 
(ii) it is also an artifact of the choice of $\alpha \in (0,1)$ since the derivative of the concave function $x^{\alpha}$ 
is decreasing in $x$. On the other hand, performance gain in the data rate for mobile users becomes smaller if $\theta$ becomes close to $1$.  When $\theta$ is small, bandwidth is allocated only to the mobile users  when they come close to the base station; however, when $\theta$ is large, there are a large number of mobile users inside the cell at any time with high probability, and hence equal bandwidth sharing among the mobile users results in significant bandwidth allocation to the mobile users which are either close or away from the base station.

 One should also note that, the amount of gain will vary depending on the topology of a cell, 
location of interfering base stations, static user locations,  shadowing realizations in 
various locations as well as fading process statistics; it is hard to quantify these effects but some intuitive conclusions can be drawn. For example, if static users are
very close to the base station, then the performance gain in the throughput of mobile users will be less since opportunistic
allocation will assign more bandwidth to static users. The numerical work presented in this section 
is only an illustration for possible performance gain by location-dependent dynamic bandwidth allocation.

\subsection{Convergence of Algorithm~\ref{algorithm:learning-algorithm-three timescale} for $\alpha=1$}
Here we consider the same network model as in Section~\ref{subsection:performance-improvement-vs-alpha} except that (i) the shadowing between any base station and any static user location or road segment center is assumed to be independent lognormal random variable with standard deviation $8$~dB, (ii) the fading gain between the origin and any location inside the cell is exponentially distributed with mean $1$ (Rayleigh fading), but the fading in any interfering link is averaged out, (iii) $\theta=0.2$, (iv) there is  only one line $x=50$ along which the mobile users traverse.

The convergence of Algorithm~\ref{algorithm:learning-algorithm-three timescale} is examined under this network setting. We first generate the network and compute the time-average expected data rate to static and mobile user classes when all users are allocated equal bandwidth in each slot; the mean data rate per slot for the static user class is then set as the target $R_0$ in \eqref{eqn:constrained_mdp_for_a_cell}, and Algorithm~\ref{algorithm:learning-algorithm-three timescale} is employed with stepsize sequences $a(t)=\frac{1}{t^{0.6}}$, $b(t)=\frac{1}{t}$ and initial estimates $R_{mobile}^{(0)}(s)=1$, $R_{static}^{(0)}(s)=1$ and $\xi^{(0)}=2$. Under Algorithm~\ref{algorithm:learning-algorithm-three timescale}, the evolution of $\frac{1}{N}\sum_{\tau=1}^N \eta_{\tau}  R_{mobile}(s(\tau))$, $\frac{1}{N}\sum_{\tau=1}^N (1-\eta_{\tau})  R_{static}(s(\tau))$ against $N$ and $\xi^{(\tau)}$ against $\tau$ for a single sample path are shown in Figure~\ref{fig:convergence-plots}. From Figure~\ref{fig:convergence-plots}, we can see that all iterates converge asymptotically, and they are close to the respective limiting values within $20000$~iterations. In practice, the convergence will be faster because the initial values  $R_{mobile}^{(0)}(s)$, $R_{static}^{(0)}(s)$ and $\xi^{(0)}$ will be chosen based on prior experience in previous days, and will be chosen close to the target values. Also, the convergence speed will depend on the network parameters, network topology, wireless propagation model and the step size sequences $\{a(\tau),b(\tau)\}_{\tau \geq 0}$. From Figure~\ref{fig:convergence-plots}, we can see that more than $150 \%$ improvement is possible in the sum throughput of mobile users while achieving the same sum throughput for static users.

\vspace{-3mm}

\section{Additional Discussion}
\label{section:implementation-issues}

\subsection{Connection Between the Cell Level Problem and a Global Problem}
\label{subsection:connection-cell-level-global-problem}
Let us consider a heterogeneous network consisting of two tiers of
base stations (BSs). The two tiers are modeled by two independent stationary, ergodic processes 
$\Phi_{macro}$ and $\Phi_{micro}$ (such as homogeneous Poisson point processes). 
 SUs are assumed to be located on $\mathbb{R}^2$   
according to a stationary, ergodic point process $\mathcal{U}_{static}$
of intensity $\lambda_{SU}$. 
We assume that MUs are moving with
constant speed $v$ along a collection of directed routes; these routes are  modeled by 
a stationary, ergodic line process (such as directed homogeneous  Poisson line process, see 
\cite[Chapter~$8$]{chiu-etal13stochastic-geometry-and-its-applications}). 
Given a realization of the  line process, we assume that, at any time $t$, two  successive MUs on any line
of $\mathcal{L}$ are separated by an exponentially distributed distance with mean $\frac{1}{\lambda_{MU}}$, 
i.e., the MUs on any line form a Poisson point 
process of intensity $\lambda_{MU}$ at any time $t$. 
Hence, the crossing of any point of a line 
by the MUs form a time homogeneous Poisson process with intensity $\lambda_{MU} v$. 

A static user is served by a macro or micro  BS, and the association rule can be arbitrary (e.g., a SU can be associated 
with the BS that sends strongest signal to the SU). Each MU is served by the nearest macro BS. We call the Voronoi cells generated by 
$\Phi_{macro}$ as macro cells. From now on, unless specified, a cell will mean a macro cell.

Note that, the heterogeneous network model is used to illustrate our model in advanced cellular network context (e.g., for LTE). 
But the analysis presented in this paper will be valid even if the network is homogeneous and each BS is allowed to serve both SUs and MUs.

Let us consider the time-slotted simplification of the above system and the problem addressed in 
Section~\ref{section:mdp-formulation-for-the-unconstrained-problem}. 
The unconstrained optimization problem \eqref{eqn:unconstrained_mdp_for_a_cell} can be used for performance optimization 
in a single macro cell. Let us enumerate the macro BSs on the plane by $\{1,2,\cdots\}$. 
Since the base stations do not communicate for making the decision on bandwidth allocation, and since each macro BS 
has different number of line segments intersecting its cell and different number of SUs associated with it, 
the dynamic bandwidth sharing policy adopted by the network is $\underline{\eta}=\times_{k=1}^{\infty} \eta^{(k)}$ where $\eta^{(k)}$ 
is the policy used by the $k$-th macro BS. Let us denote the numerator in \eqref{eqn:unconstrained_mdp_for_a_cell} for the 
$k$-th macro BS, i.e., 
$\sum_{\tau=1}^N \E _{\eta^{(k)}} \bigg( \eta_{\tau} \overline{R}_{mobile,k}(s(\tau))+\xi (1-\eta_{\tau})\overline{R}_{static,k}(s(\tau)) \bigg)$ by 
$r(k,N)$. Let us consider the following problem:

\footnotesize
\begin{eqnarray}
 \sup_{\underline{\eta}} \liminf_{M \rightarrow \infty} \liminf_{N \rightarrow \infty}   \frac{\sum_{k=1}^M r(k,N)}{NM} 
 \label{eqn:unconstrained_mdp_for_all_cells}
\end{eqnarray}
\normalsize

Now, since $\underline{\eta}=\times_{k=1}^{\infty} \eta^{(k)}$, the above problem can be rewritten as: 
\footnotesize
\begin{eqnarray*}
  \liminf_{M \rightarrow \infty} \frac{1}{M} \sum_{k=1}^M \sup_{\eta^{(k)}} \liminf_{N \rightarrow \infty}   \frac{ r(k,N)}{N}
  \label{eqn:unconstrained_mdp_for_all_cells_modified}
\end{eqnarray*}
\normalsize

Let the optimal mean reward per slot for the problem \eqref{eqn:unconstrained_mdp_for_a_cell} for cell $k$ be 
$\lambda_k$. Now, for \eqref{eqn:unconstrained_mdp_for_all_cells}, 
$\liminf_{M \rightarrow \infty} \frac{1}{M}\sum_{k=1}^M \lambda_k$ is almost surely equal to the expected optimal 
time-average reward for the typical macro cell. Hence, by solving the problem \eqref{eqn:unconstrained_mdp_for_a_cell} 
for each cell, we can maximize the expected optimal time-average reward for the typical macro cell.

\subsection{{\bf Addressing the Possibility of Error in Location Estimation for MUs}}
\label{subsection:error-in-location-estimation}

 Let us recall the framework in Section~\ref{section:mdp-formulation-for-the-unconstrained-problem}. 
It has to be noted that there can be error in estimating the location of a mobile user, and 
therefore an error in estimating the residual sojourn time of a mobile user inside a cell is possible. Let us 
assume that the error in estimation of states at any two different time slots are independent, and that we know the error 
statistics (i.e., given the observed state $\hat{s}$, we know the conditional 
distribution $p(s|\hat{s})$ of the true state $s$). Since the action in a slot 
does not affect the state transition, the best possible action one can take in a slot is to choose 

\begin{eqnarray*}
&&\eta_{\xi}^*(\hat{s})\\
&=& \arg \max_{\eta \in [0,1]}\sum_{s}p(s|\hat{s}) \bigg( \eta \overline{R}_{mobile}(s)+\xi (1-\eta) \overline{R}_{static}(s) \bigg)\\
&:=& \arg \max_{\eta \in [0,1]}\bigg( \eta \tilde{r}_{mobile}(\hat{s})+\xi (1-\eta) \tilde{r}_{static}(\hat{s}) \bigg).
\end{eqnarray*}

The structure of the optimal policy will be similar to Theorem~\ref{theorem:policy-structure}. 
Similarly, it will be optimal to work with the observed state $\hat{s}$ in case learning algorithms are employed.

\subsection{Deviation from movement along a line}
The analysis in this paper can be trivially extended to the case where the location of mobile users vary according to a positive recurrent discrete-time Markov chain over the cell divided into a finite number of area segments. The state in this case should include the location of all mobile users at a given time. However, in such cases, the location of each user needs to tracked by the base station in each slot; this is not required if the users move along straight lines with known velocity, since one can easily predict the location of a user at a time once the initial location of that user at a given time instant is known. In this paper, we considered movement of users along a line  because it stands for vehicle movements along roads and it is a simple but powerful example.

\subsection{Unequal bandwidth sharing within a single class of users}
From the anslysis presented in Section~\ref{section:mdp-formulation-for-the-unconstrained-problem}, it is clear that the decision at any given state $s$ depends only on   $\bar{R}_{mobile}(s)$ and $\bar{R}_{static}(s)$, and not on the specific bandwidth fraction allocated to each user within a user class. The percentage bandwidth allocation among mobile users within the class of mobile users determines $\bar{R}_{mobile}(s)$ which affects the decision $\eta_{\xi}^*$.

\subsection{Multiple user classes}
In case there are multiple user classes with different velocities, analogous to Theorem~\ref{theorem:policy-structure}, the optimal policy will be to allocate the entire bandwidth to a single user class at any given time. Algorithm~\ref{algorithm:learning-algorithm-three timescale} can be extended for maximizing the time average data rate for one class subject to a minimum time-average data rate constraint on each of the other classes; in this case, one Lagrange multiplier needs to be updated for each constraint, and the optimal solution for the constrained problem will involve randomization among multiple policies.

\subsection{Channel estimation versus location-based bandwidth sharing}
\label{subsection:channel-estimation-versus-location-dependent}
For fast moving users, traditional channel estimation may not be very accurate since the user might travel the fast  fading coherence distance very fast; hence, dynamic bandwidth allocation among users based on instantaneous channel qualities may not be feasible. Also, even if channel measurement is accurate, for a user moving at a velocity $72$~kmph, the channel coherence time will be less than $50$~ms; hence, gathering chennel state information from each user every $50$~ms will require huge signaling overhead. Moreover, it may be difficult to estimate the interference at any given location, since the interference at any location depends on path-loss, shadowing and time-varying fast fading gains from all interfering base stations. As an alternative, we propose location-dependent bandwidth sharing. Section~\ref{section:mdp-formulation-for-the-unconstrained-problem} deals with the situation where $\bar{R}_{mobile}(s)$ and 
$\bar{R}_{static}(s)$ are known; this amounts to assuming that the path-loss and shadowing from serving and interfering base stations are known, and the distribution of fast fading gains from the serving and interfering base stations to each location are also known. {\em We emphasize that this is an idealistic assumption, and, in practice, the serving base station has to learn $\bar{R}_{mobile}(s)$ and 
$\bar{R}_{static}(s)$ over time from the   download data volume reported by the users; this is discussed in Section~\ref{section:learning-algorithm-given-lagrange-multiplier} and Section~\ref{section:learning-algorithm-constrained-problem}. Clearly, the learning algorithms do not need any propagation based model. While  channel quality measurement, if done accurately, can result in superior user performance, our proposed algorithms for location-based bandwidth algorithms with learning are useful when accurate channel estimates and interference estimates are not available due to high velocity of users.}

Location dependent bandwidth sharing has also been discussed in \cite{bonald-mobility}, where a preference is given to the mobile users located close to the base station.

\section{Conclusion}\label{section:conclusion}
In this paper, we have  proposed and analyzed opportunistic (dynamic) bandwidth sharing depending on user location and mobility, 
in order to improve the performance of cellular networks. 
Even though we have solved the basic problem in this paper, there are numerous issues to improve upon: 
(i) In practice, there can be multiple (possibly uncountable) values of user velocity. 
Hence, a dynamic bandwidth sharing scheme that allocates bandwidth depending on exact velocity of each user 
needs to be developed (this might require classification of user velocities into a finite set).  
(ii) For general motion of users, one reasonable approach would be to divide the cell into various zones (or locations), 
and assume a  Markov evolution of user locations; similar learning techniques as in our paper  
can be applied in such situation. (iii) Testing and optimizing the proposed and subsequent algorithms in real 
data-traffic networks will be an important requirement. We propose 
to pursue these topics in our future research endeavours.

\vspace{-3mm}

{\small
\bibliographystyle{unsrt}
\bibliography{arpan-techreport}
}

\begin{figure*}[t!]
\centering{\huge{\bf Supplementary Material for ``Location Aware Opportunistic Bandwidth Sharing between Static and 
Mobile Users with Stochastic Learning in Cellular Networks''} }
\end{figure*}

\vspace{-15mm}

\begin{IEEEbiography}[{\includegraphics[width=1in,height=1in,clip,keepaspectratio]{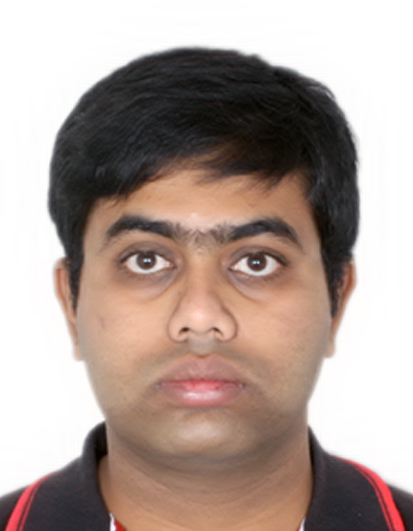}}]{Arpan 
Chattopadhyay} obtained his B.E. in Electronics and Telecommunication Engineering from Jadavpur University, 
Kolkata, India in the year 2008, and M.E. and Ph.D in Telecommunication Engineering from Indian Institute of Science, 
Bangalore, India in the year 2010 and 2015, respectively. After Ph.,D, he did his first postdoc in the group DYOGENE of INRIA, Paris, France. He is currently working in University of Southern California, Los Angeles as a postdoctoral researcher. 
His research interests include  networks, cyber-physical systems, machine learning,  and networked control.
    \end{IEEEbiography}

\vspace{-15mm}

\begin{IEEEbiography}[{\includegraphics[width=1in,height=1in,clip,keepaspectratio]{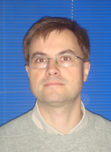}}]
{Bartlomiej Blaszczyszyn} received his PhD degree and Habilitation qualification
in applied mathematics from University of Wrocław (Poland) in 1995
and 2008, respectively. He is now a Senior Researcher at Inria (France), and
a member of the Computer Science Department of Ecole Normale Supérieure
in Paris. His professional interests are in applied probability, in particular in
stochastic modeling and performance evaluation of communication networks.
He coauthored several publications on this subject in major international
journals and conferences, as well as a two-volume book on {\em Stochastic
Geometry and Wireless Networks} NoW Publishers, jointly with F. Baccelli.
   
    \end{IEEEbiography}

    \vspace{-15mm}

   \begin{IEEEbiography}[{\includegraphics[width=1in,height=1in,clip,keepaspectratio]{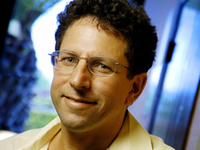}}]
   {Eitan Altman}  received the B.Sc. degree in electrical
engineering (1984), the B.A. degree in
physics (1984) and the Ph.D. degree in electrical
engineering (1990), all from the Technion-Israel
Institute, Haifa. In (1990) he further received his
B.Mus. degree in music composition in Tel-Aviv
University. Since 1990, he has been with INRIA
(National research institute in informatics and
control) in Sophia-Antipolis, France. His current
research interests include performance evaluation
and control of telecommunication networks and
in particular congestion control, wireless communications and networking
games. He is in the editorial board of several scientific journals: JEDC,
COMNET, DEDS and WICON. He has been the (co)chairman of the
program committee of several international conferences and workshops
on game theory, networking games and mobile networks.
   \end{IEEEbiography}

\renewcommand{\thesubsection}{\Alph{subsection}}

\newpage

\appendices

 \section{}
\label{appendix}

\subsection{Proof of Lemma~\ref{lemma:mdp-optimal-rate-of-static-mobile-users-increasing-decreasing-in-xi}}
\label{subsection:proof-of-rate-of-mobile-users-decreasing-in-xi}
Let $\xi\geq 0$ and $\kappa>0$. By the optimality of $\eta_{\xi}^*(\cdot)$ and $\eta_{\xi+\kappa}^*(\cdot)$, we can write:
\small
$$ \overline{R}_{mobile}^*(\xi)+ \xi \overline{R}_{static}^*(\xi) \geq \overline{R}_{mobile}^*(\xi+\kappa)+ \xi \overline{R}_{static}^*(\xi+\kappa) $$
and 
\small
$$ \overline{R}_{mobile}^*(\xi+\kappa)+ (\xi+\kappa) \overline{R}_{static}^*(\xi+\kappa) \geq \overline{R}_{mobile}^*(\xi)+ (\xi+\kappa) \overline{R}_{static}^*(\xi) $$
 \normalsize
 By adding these two equations, we obtain: $$\overline{R}_{static}^*(\xi+\kappa) \geq \overline{R}_{static}^*(\xi).$$  
 Similarly we can prove that $\overline{R}_{mobile}^*(\xi)$ decreases in $\xi$.

 \subsection{Proof of Theorem~\ref{theorem:single-timescale-convergence}}
 \label{subsection:proof-of-single-timescale-convergence-theorem}
 Let us rewrite the update equation in Algorithm~\ref{algorithm:learning-algorithm-single-timescale} as follows:
 
 \footnotesize
   \begin{eqnarray*}
  R_{mobile}^{(\tau+1)}(s)&=&R_{mobile}^{(\tau)}(s)+a(\nu(s,1,\tau))\ind \{s(\tau)=s\} \ind \{\eta_{\tau}=1\}\\
  && \times \bigg(\overline{R}_{mobile}(s)-R_{mobile}^{(\tau)}(s)+N^{(\tau+1)}(s,1) \bigg)\\
    R_{static}^{(\tau+1)}(s)&=&R_{static}^{(\tau)}(s)+a(\nu(s,0,\tau))\ind \{s(\tau)=s\} \ind \{\eta_{\tau}=0\}\\
  && \times \bigg(\overline{R}_{static}(s)-R_{static}^{(\tau)}(s)+N^{(\tau+1)}(s,0) \bigg)\label{eqn:single-timescale-update-equation-in-new-form}
 \end{eqnarray*}
 \normalsize
 where $$N^{(\tau+1)}(s,1):=R_{mobile}(s)-\overline{R}_{mobile}(s),$$ and $$N^{(\tau+1)}(s,0):=R_{static}(s)-\overline{R}_{static}(s).$$ 
 This is an asynchronous stochastic approximation iteration as described in \cite{borkar08stochastic-approximation-book}, 
 with $N(s,1)$ and $N(s,0)$ as Martingale difference noise sequences. However, for each $s$,  
 $$\lim \inf_{\tau \rightarrow \infty}\frac{\nu(s,1,\tau)}{\tau} \geq \frac{g(s)\epsilon}{2}>0, $$  where $g(s)$ has been defined in 
 Section~\ref{subsection:policy-structure-and-computation}, and  
$$\lim \inf_{\tau \rightarrow \infty}\frac{\nu(s,0,\tau)}{\tau} \geq \frac{g(s)\epsilon}{2}>0 $$ almost surely. 

Since each iterate is updated infinitely often, and since the iterations of various components of the iterates are uncoupled, 
for each $s$ we can cast the iteration as an ordinary stochastic approximation as defined in 
\cite[Chapter~$2$]{borkar08stochastic-approximation-book}.
 
 Now we will check some conditions from  \cite{borkar08stochastic-approximation-book}. 
 Let us denote $\underline{R}:=\{R_{mobile}(s), R_{static}(s) \}_{\forall s}$.
 
 \textbf{Checking Assumption~$(A1)$ of \cite[Chapter~$2$]{borkar08stochastic-approximation-book}:} Clearly, 
 $\overline{R}_{mobile}(s)-R_{mobile}^{(\tau)}(s)$ is Lipschitz in $R_{mobile}^{(\tau)}(s)$ and 
$\overline{R}_{static}(s)-R_{static}^{(\tau)}(s)$ is Lipschitz in $R_{static}^{(\tau)}(s)$ 
for each $s$; hence, this assumption is satisfied.
 
\textbf{Checking Assumption~$(A2)$ of \cite[Chapter~$2$]{borkar08stochastic-approximation-book}:} This 
assumption is satisfied by the choice of the step size sequence.
 
\textbf{Checking Assumption~$(A3)$ of \cite[Chapter~$2$]{borkar08stochastic-approximation-book}:} It is easy to see that, 
$\{N^{(\tau+1)}(s,1), N^{(\tau+1)}(s,0)\}_{\tau \geq 1}$ for each $s$ is a sequence of Martingale 
difference noise with zero mean, adapted to the sigma algebra generated by 
$\{N^{(k)}(s,1), N^{(k)}(s,0)\}_{0 \leq k \leq \tau, \forall s}$. Also, 
the conditional mean of $|N^{(\tau+1)}(s,1)|^2$ given all the noise values up to time $\tau$ is uniformly 
upper bounded by some constant, since $T$ is Geometrically distributed and fading process is bounded by 
Assumption~\ref{assumption:fading-ergodic}. Hence, 
Assumption~$(A3)$ of \cite[Chapter~$2$]{borkar08stochastic-approximation-book} is satisfied. 

\textbf{Checking Assumption~$(A5)$ of \cite[Chapter~$3$]{borkar08stochastic-approximation-book}:}
Note that, $\lim_{c \rightarrow \infty} \frac{\overline{R}_{mobile}(s)-cx(s,1)}{c}=-x(s,1)$ 
is continuous in $x(s,1)$ for all $s$. 
Also, $\frac{\overline{R}_{mobile}(s)-cx(s,1)}{c}$ is decreasing in $c$. Hence, 
by Theorem~$7.13$ of \cite{rudin76principles-of-mathematical-analysis}, 
convergence of $\frac{\overline{R}_{mobile}(s)-cx(s,1)}{c}$ 
over compacts is uniform. 
Also, the collection of  ODEs of the form 
$\dot{x}(s,1) =  \lim_{c \rightarrow \infty} \frac{\overline{R}_{mobile}(s)-cx(s,1)}{c}=- x(s,1)$   
has a unique unique globally asymptotically stable equilibrium 
$x(s,1)=x(s,0)=0$ for all $s$. Hence, this assumption is satisfied.

Let us consider the following ODE for all $s$:
\begin{eqnarray*}
 \dot{x}(s,1) &=& \overline{R}_{mobile}(s)-x(s,1)  \\
 \dot{x}(s,0) &=&  \overline{R}_{static}(s)-x(s,0)  
\end{eqnarray*}
The above ODE 
has a unique globally asymptotically stable equilibrium $x(s,1)=\overline{R}_{mobile}(s)$ and 
$x(s,0)=\overline{R}_{static}(s)$. Hence, by 
\cite[Theorem~$7$, Chapter~$3$]{borkar08stochastic-approximation-book} and 
\cite[Theorem~$2$, Chapter~$2$]{borkar08stochastic-approximation-book}, convergence of 
Algorithm~\ref{algorithm:learning-algorithm-single-timescale} follows.

Now, it is easy to see that, 
$$|\lambda_{\epsilon}^*(\xi)-\lambda^*(\xi)| \leq \frac{\epsilon}{2 }\sum_{s}g(s)  \E  |R_{mobile}(s)-R_{static}(s) |$$.
 
The second part of the theorem follows from this.\qed

\subsection{Proof of Theorem~\ref{theorem:convergence-two-timescale-iteration}}
\label{subsection:proof-of-three timescale-convergence-theorem}
We will prove desired convergence in the two timescales one by one.

\subsubsection{Convergence in the faster timescale}
\label{subsubsection:convergence-in-fastest-timescale}
\begin{lemma}\label{lemma:convergence-of-rate-estimates-three timescale}
 Under Algorithm~\ref{algorithm:learning-algorithm-three timescale}, we have 
 $\lim_{\tau \rightarrow \infty} R_{mobile}^{(\tau)}(s)=\overline{R}_{mobile}(s)$ and  
 $\lim_{\tau \rightarrow \infty} R_{static}^{(\tau)}(s)=\overline{R}_{static}(s)$ for all $s$ almost surely.
\end{lemma}
\begin{proof}
 The proof is similar to the proof of Theorem~\ref{theorem:single-timescale-convergence}.
\end{proof}

\begin{lemma}\label{lemma:convergence-of-p-iterates}
 Under Algorithm~\ref{algorithm:learning-algorithm-three timescale}, we have 
 $\lim_{\tau \rightarrow \infty}|p^{(\tau)}-p_{\epsilon}^*(\xi^{(\tau)})|=0$. 
\end{lemma}
\begin{proof}
Note that, we  can rewrite the $p$ iteration (using Taylor series expansion) as follows:
\tiny
 \begin{eqnarray*}
  && p^{(\tau+1)}=\bigg[ p^{(\tau)}+a(\tau) (\sum_s \ind \{s(\tau)=s,\eta_{\tau}=0\} \\
  && \times R_{static}(s)-R_0 ) \bigg]_0^1 \\
  &=& p^{(\tau)}+o( a(\tau) ) \\
  && + \lim_{\beta \rightarrow \infty} \frac{  \bigg[ p^{(\tau)}+\beta (\sum_s \ind \{s(\tau)=s,\eta_{\tau}=0\} R_{static}(s)-R_0 ) \bigg]_0^1 - p^{(\tau)} }{\beta}\\
 &=& p^{(\tau)}+o( a(\tau) ) + N^{(\tau)} \\
  && + a(\tau) \E \lim_{\beta \rightarrow \infty} \frac{  \bigg[ p^{(\tau)}+\beta (\sum_s \ind \{s(\tau)=s,\eta_{\tau}=0\} R_{static}(s)-R_0 ) \bigg]_0^1 - p^{(\tau)} }{\beta}
 \end{eqnarray*}
\normalsize

where $N^{(\tau)}$ is a Martingale difference noise sequence, $o(b(\tau))$ is the tail of the Taylor 
series expansion, and  the expectation is under the randomized 
policy $\eta^{(\epsilon)}(\cdot | \cdot,\cdot,\cdot,\cdot)$.

 Now,
\begin{eqnarray*}
&& \Pro (\eta_{\tau}=0|s(\tau)=s,\mathcal{R}_{\tau},\xi^{(\tau)},p^{(\tau)}) \\
&=&\frac{\epsilon}{2}+(1-\epsilon)\Pro (R_{mobile}^{(\tau)}(s)-(\xi^{(\tau)}+\Delta_{\tau}) R_{static}^{(\tau)}(s) \leq 0 )
\end{eqnarray*}

Note that, $ \Pro (\eta_{\tau}=0|s(\tau)=s,\mathcal{R}_{\tau},\xi^{(\tau)},p^{(\tau)})$ is continuous in $(R_{mobile}^{(\tau)}(s), R_{static}^{(\tau)}(s), p^{(\tau)}, \xi^{(\tau)})$. As a result of this and 
Lemma~\ref{lemma:convergence-of-rate-estimates-three timescale}, the difference of the above expectation under 
the policies $\eta^{(\epsilon)}(\cdot | \cdot,\cdot,\cdot,\cdot)$ and 
$\eta^{(\epsilon)}(\cdot | \cdot,\{\overline{R}_{mobile}(s),\overline{R}_{static}(s)\}_{\forall s},\cdot,\cdot)$ go to $0$ as 
$\tau \rightarrow \infty$.

Now we claim that $(p^{(\tau)},\xi^{(\tau)})$ converges 
to the internally chain transitive invariant sets of the o.d.e. 
\tiny
\begin{eqnarray*}
&& \dot{p}(t)\\
&=& \E \lim_{\beta \rightarrow 0}\frac{[p(t)+\beta  (\sum_s \ind \{s(t)=s,\eta_{t}=0\} 
R_{static}(s)-R_0 ) ]_0^1-p(t)}{\beta}, \\
&&\dot{\xi}(t)=0,
\end{eqnarray*} 
\normalsize
where 
the expectation is under the policy 
$\eta^{(\epsilon)}(\cdot | \cdot,\{\overline{R}_{mobile}(s),\overline{R}_{static}(s)\}_{\forall s},\cdot,\cdot)$ and the 
fading distribution. 

Note that, this o.d.e becomes $\dot{p}(t) \geq 0$ at $p(t)=0$, $\dot{p}(t) \leq 0$ at $p(t)=1$,  
and else 

\begin{eqnarray*}
\dot{p}(t)&=&\E  \bigg( \sum_s \ind \{s(t)=s,\eta_{t}=0\} R_{static}(s)-R_0 ) \bigg) \\
&=& \sum_s g(s)\Pro (\eta_t(s)=0) \overline{R}_{static}(s)-R_0.
\end{eqnarray*}  

Since $\Pro (\eta(s)=0)$ is decreasing in $p(t)$, the o.d.e. 
\footnotesize
\begin{eqnarray*}
&&\dot{p}(t)\\
&=& \E \lim_{\beta \rightarrow 0}\frac{[p(t)+\beta  (\sum_s \ind \{s(t)=s,\eta_{t}=0\} 
R_{static}(s)-R_0 ) ]_0^1-p(t)}{\beta}
\end{eqnarray*}
\normalsize

 can have at most one limit point. This limit point, which we call 
$p_{\epsilon}^*(\xi)$, is either in $\{0,1\}$ or it is a stationary point of the above o.d.e.

Also, by an argument similar to Lemma~\ref{lemma:optimal-p-lipschitz-in-xi}, 
$p_{\epsilon}^*(\xi)$ is Lipschitz continuous in $\xi$.

Hence, using an argument similar to \cite[Chapter~$6$, Lemma~$1$]{borkar08stochastic-approximation-book}, we prove 
the lemma.
\end{proof}

\subsubsection{Convergence in the slower timescale}
\label{subsubsection:convergence-in-slowest-timescale}

We first prove the following lemma.

\begin{lemma}\label{lemma:probabilities-continuous-in-xi}
 $\Pro (\eta(s)=0)$ under the randomized policy 
$\eta^{(\epsilon)}(\cdot | s,\{\overline{R}_{mobile}(s),\overline{R}_{static}(s)\}_{\forall s},p_{\epsilon}^*(\xi),\xi)$ 
is continuous in $\xi$.
\end{lemma}
\begin{proof}
We have:
\begin{eqnarray*}
 && \Pro (\eta(s)=0)\\
 &=& \frac{\epsilon}{2}+(1-\epsilon)\Pro \bigg( \overline{R}_{mobile}(s) -(\xi+\Delta) \overline{R}_{static}(s) \leq 0 \bigg)\\
  &=& \frac{\epsilon}{2}+(1-\epsilon)\Pro \bigg(\Delta \geq \frac{\overline{R}_{mobile}(s)}{\overline{R}_{static}(s)} -\xi  \bigg)\\
\end{eqnarray*}
This is continuous in $\xi$ and $p_{\epsilon}^*(\xi)$, and by an argument similar to Lemma~\ref{lemma:optimal-p-lipschitz-in-xi}, 
$p_{\epsilon}^*(\xi)$ is continuous in $\xi$. This proves the lemma.
 \end{proof}

Now we state the final lemma.

\begin{lemma}\label{lemma:convergence-of-xi-iterates}
 Almost surely, as $\tau \rightarrow \infty$, the iterates $\xi^{(\tau)}$ converges to the projection of 
 $\mathcal{K}_{\epsilon}(R_0)$ onto the $\xi$ axis.
\end{lemma}
\begin{proof}
 The proof follows using similar arguments as 
 \cite[Appendix~E.C.3, Appendix~E.C.4 and Appendix~E.C.5]{chattopadhyay-etal15measurement-based-impromptu-deployment-arxiv-v1}. 
 The arguments require the results in Lemma~\ref{lemma:probabilities-continuous-in-xi}, 
 Lemma~\ref{lemma:convergence-of-rate-estimates-three timescale} and 
 Lemma~\ref{lemma:convergence-of-p-iterates}. The choice of $A$ should be sufficiently large, otherwise 
 $\xi^{(\tau)}$ might converge to $A$. 
 
 The proof of the slowest timescale convergence in 
 \cite[Theorem~$12$]{chattopadhyay-etal15measurement-based-impromptu-deployment-arxiv-v1} involves  checking of  
 five conditions required for \cite[Theorem~$5.3.1$]{kushner-clark78SA-constrained-unconstrained}.  
 The constrained MDP associated with 
 \cite[Theorem~$12$]{chattopadhyay-etal15measurement-based-impromptu-deployment-arxiv-v1} had two constraints and hence 
 two slower timescale iterates, whereas we have only one slowest timescale iterate for which it is 
 easier to check these five conditions. Since these conditions hold, we can claim that the $\xi$ iteration  
 converges almost surely to the set of stationary points of a suitable o.d.e. Again, since we have only one slower 
 timescale iterate, using large enough $A$ is sufficient to ensure that the stationary points of that o.d.e. 
 lie in $(0,A)$; it was much more complicated in \cite[Theorem~$12$]{chattopadhyay-etal15measurement-based-impromptu-deployment-arxiv-v1} 
 since there were two slower timescale iterates.
 
 The proof of this lemma requires Lemma~\ref{lemma:probabilities-continuous-in-xi} and the fact that $R_{static}(s)$ has a bounded support 
 (since by Assumption~\ref{assumption:fading-ergodic}, fading gain distribution has bounded support).
\end{proof}
In light of Lemma~\ref{lemma:convergence-of-rate-estimates-three timescale}, 
Lemma~\ref{lemma:convergence-of-p-iterates} and Lemma~\ref{lemma:convergence-of-xi-iterates}, 
the theorem is proved.\qed

\subsection{Proof of Corollary~\ref{corollary:two-timescale-learning-optimal-for-constrained-problem}}
\label{subsection:proof-of-corollary}
Using arguments similar to the proof of Lemma~\ref{lemma:probabilities-continuous-in-xi},  under 
$\eta^{(\epsilon)}(\cdot | \cdot,\cdot,\cdot,\cdot)$, one can claim that 
$$ \Pro(\eta_{\tau}=0|s(\tau)=s,\mathcal{R}_{\tau}=\mathcal{R},
\xi^{(\tau)}=\xi,p^{(\tau)}=p,\epsilon)$$  
is continuous in $(\mathcal{R},\xi,p,\epsilon)$ for given $s$. Hence, by 
Theorem~\ref{theorem:convergence-two-timescale-iteration}, we can claim that: 
$\Pro(\eta_{\tau}=0|s(\tau)=s,\epsilon)$ converges to the set 
$ \{x: x=\eta^{(\epsilon)}(0|s,\{\overline{R}_{mobile}(s),\overline{R}_{static}(s)\}_{\forall s}, \xi,p ), (\xi,p) \in 
\mathcal{K}_{\epsilon}(R_0)\}$ almost surely as $\tau \rightarrow \infty$. 

Now, 
\small
\begin{eqnarray*}
&& \lim_{\epsilon \downarrow 0} \{x: x=\eta^{(\epsilon)}(0|s,\{\overline{R}_{mobile}(s),\overline{R}_{static}(s)\}_{\forall s}, \xi,p ), (\xi,p) \in 
\mathcal{K}_{\epsilon}(R_0)\} \\
&=&\{x: x=\eta^{(0)}(0|s,\{\overline{R}_{mobile}(s),\overline{R}_{static}(s)\}_{\forall s}, \xi,p ), (\xi,p) \in 
\mathcal{K}(R_0)\}. 
\end{eqnarray*}
\normalsize

The proof trivially follows from this.\qed

\end{document}